%% file: main.tex
\documentclass[acmsmall, nonacm, screen]{acmart}

\input{src/preamble.tex}

\begin{document}

\title{Membership Testing for Semantic Regular Expressions}

\author{Yifei Huang}
\orcid{0009-0006-5675-4065}
\affiliation{
  \institution{University of Southern California}
  \city{Los Angeles}
  \state{CA}
  \country{USA}
}
\email{yifeih@usc.edu}
\author{Matin Amini}
\orcid{0009-0006-8776-3925}
\affiliation{
  \institution{University of Southern California}
  \city{Los Angeles}
  \state{CA}
  \country{USA}
}
\email{matinami@usc.edu}
\author{Alexis Le~Glaunec}
\orcid{0000-0002-5444-5924}
\affiliation{
  \institution{Rice University}
  \city{Houston}
  \state{TX}
  \country{USA}
}
\email{afl5@rice.edu}
\author{Konstantinos Mamouras}
\orcid{0000-0003-1209-7738}
\affiliation{
  \institution{Rice University}
  \city{Houston}
  \state{TX}
  \country{USA}
}
\email{mamouras@rice.edu}
\author{Mukund Raghothaman}
\orcid{0000-0003-2879-0932}
\affiliation{
  \institution{University of Southern California}
  \city{Los Angeles}
  \state{CA}
  \country{USA}
}
\email{raghotha@usc.edu}






\if\jobis{app} \else
\input{src/abstract.tex}
\maketitle

\mclearpage
\input{src/intro.tex}        \mclearpage
\input{src/semre.tex}        \mclearpage
\input{src/alg.tex}          \mclearpage
\input{src/complexity.tex}   \mclearpage
\input{src/eval.tex}         \mclearpage
\input{src/related.tex}      \mclearpage
\input{src/concl.tex}        \mclearpage
\input{src/artifact.tex}     \mclearpage
\input{src/acks.tex}         \mclearpage


\mclearpage
\bibliographystyle{ACM-Reference-Format}
\citestyle{acmnumeric} 
\bibliography{src/references.bib, src/refs_km.bib}
\fi



\if\jobis{short} \else
\clearpage
\appendix
\input{app/alg.tex}          \mclearpage
\input{app/complexity.tex}   \mclearpage
\input{app/eval.tex}         \mclearpage
\fi

\end{document}

%% file: src/preamble.tex
\usepackage{ifpdf}
\ifpdftex
  \usepackage[T1]{fontenc}
  \usepackage[utf8]{inputenc}
\fi

\usepackage{cancel}                    
\usepackage{mathpartir}                
\usepackage{mathtools}                 
\usepackage{microtype}
\usepackage{paralist, enumitem}
\usepackage{subcaption}
\usepackage{tikz}

\usetikzlibrary{
  automata, backgrounds, calc,
  decorations.pathmorphing, decorations.pathreplacing, decorations.shapes,
  fit, positioning
}

\AtEndPreamble{%
  \theoremstyle{acmdefinition}
  \newtheorem{note}[theorem]{Note}
  \newtheorem{remark}[theorem]{Remark}
  \newtheorem{assumption}[theorem]{Assumption}}

\usepackage{xr}
\usepackage{thmtools, reptheorem}      
\def\StripPrefix#1>{}
\def\jobis#1{FF\fi
  \def\predicate{#1}%
  \edef\predicate{\expandafter\StripPrefix\meaning\predicate}%
  \edef\job{\jobname}%
  \ifx\job\predicate
}

\if\jobis{short}
  \newcommand{\refshort}[1]{\ref{#1}}
  \newcommand{\refapp}[1]{\ref{full:#1}}
\else
\if\jobis{app}
  \newcommand{\refshort}[1]{\ref{full:#1}}
  \newcommand{\refapp}[1]{\ref{#1}}
  \loadtheorems{main.thm}
\else
  \newcommand{\refshort}[1]{\ref{#1}}
  \newcommand{\refapp}[1]{\ref{#1}}
  \theoremfile
\fi
\fi

\newcommand{\Bool}{\mathrm{Bool}}
\newcommand{\true}{\mathrm{true}}
\newcommand{\false}{\mathrm{false}}

\newcommand{\Lor}{\operatorname{\bigvee}}

\newcommand{\union}{\cup}
\newcommand{\Union}{\bigcup}
\newcommand{\intersection}{\cap}

\let\originalmiddle=\middle
\def\middle#1{\mathrel{}\originalmiddle#1\mathrel{}}

\newcommand{\strempty}{\epsilon}

\ifluatex
  \usepackage{emoji}
  \newcommand{\oracle}{\text{\emoji{crystal-ball}}}
\fi
\ifpdftex
  \usepackage{hwemoji}
  \newcommand{\oracle}{🔮}
\fi

\newcommand{\query}[1]{\langle #1 \rangle}
\newcommand{\interp}[1]{\llbracket #1 \rrbracket}
\newcommand{\queryp}[1]{[ #1 ]}
\newcommand{\skel}{\operatorname{skel}}

\newcommand{\open}{\operatorname{open}}
\newcommand{\close}{\operatorname{close}}
\newcommand{\blank}{\operatorname{blank}}
\newcommand{\str}{\operatorname{str}}

\newcommand{\scalefactor}{0.65}

\newcommand{\qgstart}{\text{start}}
\newcommand{\qgend}{\text{end}}
\newcommand{\qgidx}{\operatorname{idx}}
\newcommand{\qgat}{\mathrel{:}} 
\tikzset{qgstart/.style={rectangle, rounded corners=2mm, draw}}
\tikzset{qgend/.style={rectangle, rounded corners=2mm, draw}}
\tikzset{qgblank/.style={circle, draw=blue!50, fill=blue!20}}

\newcommand{\qcon}{\operatorname{qcon}}
\newcommand{\efp}{\to^{\strempty^*}_F}

\newcommand{\Gad}{\text{Gad}}

\newcommand{\alive}{\operatorname{Alive}}

\newcommand{\backref}{\operatorname{Backref}}
\newcommand{\loq}{\operatorname{LOQ}}
\newcommand{\aq}{\operatorname{Matched}}

\newcommand{\grepo}{\text{grep}^O}
\newcommand{\padded}[1]{\Sigma^*r_\text{#1}\Sigma^*}

\newcommand{\boldsmall}{$\mathbf{\le 0.1}$}

\newcommand{\mclearpage}{}

%% file: src/abstract.tex
\begin{abstract}
This paper is about semantic regular expressions (SemREs). This is a concept that was recently
proposed by~\citet{Smore} in which classical regular expressions are extended with a primitive to
query external oracles such as databases and large language models (LLMs). SemREs can be used to
identify lines of text containing references to semantic concepts such as cities, celebrities,
political entities, etc. The focus in their paper was on automatically synthesizing semantic regular
expressions from positive and negative examples.
In this paper, we study the \emph{membership testing problem}:
\begin{enumerate}
\item We present a two-pass NFA-based algorithm to determine whether a string $w$ matches a SemRE
  $r$ in $O(|r|^2 |w|^2 + |r| |w|^3)$ time, assuming the oracle responds to each query in unit time.
  In common situations, where oracle queries are not nested, we show that this procedure runs in
  $O(|r|^2 |w|^2)$ time.
  %
  Experiments with a prototype implementation of this algorithm validate our theoretical analysis,
  and show that the procedure massively outperforms a dynamic programming-based baseline, and incurs
  a $\approx 2 \times$ overhead over the time needed for interaction with the oracle.
\item We establish connections between SemRE membership testing and the triangle finding problem
  from graph theory, which suggest that developing algorithms which are simultaneously practical and
  asymptotically faster might be challenging.
  %
  Furthermore, algorithms for classical regular expressions primarily aim to optimize their time and
  memory consumption. In contrast, an important consideration in our setting is to minimize the cost
  of invoking the oracle. We demonstrate an $\Omega(|w|^2)$ lower bound on the number of oracle
  queries necessary to make this determination.
\end{enumerate}
\end{abstract}

%% file: src/intro.tex
\section{Introduction}
\label{sec:intro}


Since their introduction in the 1970s, regular expression matching engines such as grep~\cite{grep}
and sed~\cite{sed} have emerged as the workhorses of data extraction and transformation workflows.
Regular expressions are well-suited to describe simple syntactic structures in text, such as for
identifying sequences of tokens and for matching delimiter boundaries. Apart from their ease of use,
they are backed by well-understood theoretical foundations and efficient matching algorithms.

\citet{Smore} recently proposed extending classical regular expressions with primitives to query
large language models (LLMs). For example, one might search for mentions of Turing award winners in
some text using the pattern $\Sigma^* \query{\text{Turing award winner}} \Sigma^*$, assuming an
oracle exists which can identify their names.
One might also conceivably apply such patterns to search for mentions of Eastern European cities,
respiratory diseases endemic among humans, to flag passwords and SSH keys accidentally left behind
in source code commits, or to query external services such as DNS and WHOIS while categorizing spam
email. The motivation is that classical regular expressions are a bad tool to describe such semantic
categories, which are better resolved by asking a search engine, an LLM, or by querying other
external sources of information such as databases of award winners or atlases of major cities.

The major focus in~\cite{Smore} was on synthesizing semantic regular expressions from user-provided
examples of positive and negative data. We assert that a more fundamental problem involves
\emph{membership testing}: Does a string $w$ match a given semantic regular expression $r$? This
problem will be our subject of study in this paper.


In contrast to~\cite{Smore}, which uses an expansive formalism (including regular expression
operators such as intersection and complement), we will focus on a simplified core subset of SemREs
(oracle refinement + the classical regular expression operators---union, concatenation, and
Kleene-*). This allows us to isolate the impact of oracle refinement on the complexity of membership
testing.

Still, even for this simplified fragment, it is relatively straightforward to show an
$\Omega(|w|^2)$ lower bound on the number of oracle queries needed (and hence on the worst-case
running time) to make this determination. See Theorem~\ref{thm:complexity:quad}. It follows that any
automata-based matching algorithm---which are frequently associated with linear time algorithms---%
must necessarily perform non-trivial additional work.

In this context, our primary contribution is a two-pass NFA-based matching algorithm: The first pass
recognizes syntactic patterns required by the underlying classical regular expression, and
constructs a data structure which we call a query graph. The query graph encodes outstanding query-%
substring pairs that need to submitted to the oracle for verification and Boolean relationships
among the results of these queries. The second pass uses dynamic programming to evaluate the query
graph and compute the final result.

The principal challenge is in designing the query graph so that it can both be rapidly constructed
and efficiently evaluated. One particularly tricky part of this procedure occurs while constructing
the query graph, when we must track the string indices at which oracle queries begin and end. Recall
that because classical Thompson NFAs also contain $\strempty$-transitions, their evaluation requires
computing their transitive closure after processing each subsequent character of the input string.
Unfortunately, in our case, these $\strempty$-closures might themselves involve complicated query
structures that need to be carefully processed. Our solution is a novel gadget which summarizes all
possible patterns in which substrings may need to be examined by the oracle using only $\Theta(|r|)$
vertices in the query graph.
At this point, query graph construction reduces to a relatively straightforward tiling of this
gadget after processing each character. Taken together, our end-to-end membership testing algorithm
runs in $O(|r|^2 |w|^2 + |r| |w|^3)$ time, while making $O(|r| |w|^2)$ oracle queries.


In Section~\ref{sub:complexity:triangles}, we show that SemRE matching is at least as hard as
finding triangles in graphs. While subcubic (i.e., $O(n^{3 - \epsilon})$ time) algorithms are known
for this latter problem~\cite{Williams2:FOCS2010}, they all build on algorithms for fast (integer)
matrix multiplication, and are consequently slower in practice. This suggests that matching
algorithms which are both practical and asymptotically faster might be challenging to devise.

This hardness result fundamentally relies on SemREs with nested queries. An example of such a
pattern is $(\Sigma^* (\Sigma^* \land \query{\text{City}}) \Sigma^*) \land
\query{\text{Celebrity}}$, which identifies celebrities whose names also include the names of
cities, such as \texttt{Paris Hilton}. We expect that common SemREs will not utilize such nested
queries. In their absence, we can show that our algorithm runs in $O(|r|^2 |w|^2)$ time.

We also include the results of an experimental evaluation: We use a benchmark dataset consisting of
two text files and nine SemREs, and compare our matching algorithm to a simple dynamic
programming-based baseline. Our algorithm turns out to have $101 \times$ the throughput of the
baseline, and a $2\times$ overhead over the cost of consulting the oracle. Furthermore, our
algorithm makes 51\% fewer oracle queries than the baseline. Given the costs of large-scale LLM
use~\cite{OpenAI:Pricing}, our hope is that this algorithm will enable the practical application of
oracle-backed regular expressions.


\paragraph{Summary of Contributions.}

To summarize, our paper makes the following contributions:
\begin{enumerate}
\item We initiate the study of the membership testing problem for semantic regular expressions
  (SemREs), and present an automaton-based matching algorithm that runs in $O(|r|^2 |w|^2 +
  |r| |w|^3)$ time with $O(|r| |w|^2)$ queries. Additional assumptions, such as the absence of
  nesting, permit faster performance guarantees to be made. 
\item We report on a set of SemREs and corresponding text files that can be used to benchmark
  membership testing algorithms for SemREs. Experiments with this dataset validate our theoretical
  analysis, and indicate that our algorithm outperforms a simple dynamic-programming based baseline.
\item We establish the first lower bounds on SemRE membership testing:
  \begin{inparaenum}[(\itshape a\upshape)]
  \item That at least $\Omega(|w|^2)$ oracle queries are necessary in the worst case, and
  \item that subcubic algorithms for membership testing could be used to find triangles in graphs in
    subcubic time.
  \end{inparaenum}
\end{enumerate}

%% file: src/semre.tex
\section{Semantic Regular Expressions}
\label{sec:semre}

\input{src/semre-formalism.tex} \mclearpage
\input{src/semre-examples.tex}

%% file: src/semre-formalism.tex
\subsection{The Formalism}
\label{sub:semre:formalism}

Fix a finite alphabet $\Sigma$, a space of possible queries $Q$, and let $\oracle : Q \times
\Sigma^* \to \Bool$ be the external oracle. We choose this symbol to evoke the idea of the oracle
being a magical crystal ball. We define \emph{semantic regular expressions} (SemRE) as the
productions of the following grammar:
\begin{alignat}{1}
  r & \Coloneqq \bot \mid
                \strempty \mid
                a \mid
                r_1 + r_2 \mid
                r_1 r_2 \mid
                r^* \mid
                r \land \query{q},
\label{eq:semre:syntax}
\end{alignat}
for $a \in \Sigma$ and $q \in Q$. Observe that the only difference from classical regular
expressions~\cite{Sipser} is the final production rule $r \Coloneqq \dots \mid r \land \query{q}$.
Each semantic regular expression $r$ identifies a set of strings $\interp{r} \subseteq \Sigma^*$,
defined recursively as follows:
\begin{equation}
\left.
\begin{aligned}
  \interp{\bot}              & = \emptyset,                                         \\
  \interp{\strempty}         & = \{ \strempty \},                                   \\
  \interp{a}                 & = \{ a \},                                           \\
  \interp{r_1 + r_2}         & = \interp{r_1} \union \interp{r_2},                  \\
  \interp{r_1 r_2}           & = \interp{r_1} \interp{r_2},                         \\
  \interp{r^*}               & = \Union_{i \geq 0} \interp{r}^i, \text{ and}        \\
  \interp{r \land \query{q}} & = \{ w \in \interp{r} \mid \oracle(q, w) = \true \},
\end{aligned}
\qquad \right\}
\label{eq:semre:semantics}
\end{equation}
where the concatenation and exponentiation of sets of strings $A, B \subseteq \Sigma^*$ are defined
as usual:
\begin{alignat*}{1}
  AB        & = \{ w_1 w_2 \mid w_1 \in A, w_2 \in B \}, \\
  A^0       & = \{ \strempty \}, \text{ and} \\
  A^{n + 1} & = A^n A, \text{ for } n \geq 0 \text{ respectively}.
\end{alignat*}
Once again, the only difference from classical language theory is in the definition $\interp{r \land
\query{q}}$ of oracle queries. Given a SemRE $r$ and a string $w \in \Sigma^*$, the central problem
of this paper is to determine whether $w \in \interp{r}$.
An easy solution is to use dynamic programming and memoization to operationalize Equation~%
\ref{eq:semre:semantics}. This is the approach used in the SMORE implementation by~\citet{Smore}.%
\footnote{See Lines 271--394 of the file \texttt{Smore-main/lib/interpreter/executor.py} in the
  artifact that may be downloaded from \url{https://doi.org/10.5281/zenodo.8144182}.} This algorithm
would require $O(|r||w|^3)$ time in the worst case. \emph{Can we do better?}

%% file: src/semre-examples.tex
\subsection{Examples}
\label{sub:semre:examples}

In common use, one might set $\Sigma$ to be the set of ASCII or Unicode code points. One might
choose $Q = \{ \text{Sportsperson}, \text{Scientist}, \text{Eastern European city}, \dots \}$. Of
course, $\oracle$ would depend on the backing oracle, but for the sake of concreteness, one might
assume $\oracle(\text{Sportsperson}, w) = \true$ for $w \in \{ \texttt{Simone Biles},
\texttt{Lionel Messi}, \texttt{Roger Federer}, \dots\}$, and that $\oracle(\text{Scientist}, w) =
\true$ for $w \in \{ \texttt{Albert Einstein}, \texttt{Marie Curie}, \texttt{Charles Darwin},
\dots \}$.

We start with two simple examples of semantic regular expressions: Lines of text containing
political news might match the expression $\Sigma^* (\Sigma^* \land \query{\text{Politician}})
\Sigma^*$, while rosters of sports teams might be described by the pattern $((\Sigma^* \land
\query{\text{Sportsperson}}) \texttt{", "})^* (\Sigma^* \land \query{\text{Sportsperson}})$.

\begin{note}
Observe the particularly common idiom, $\Sigma^* \land \query{q}$. When the intent is clear, we will
shorten this and simply write $\query{q}$. For example, we will write $\Sigma^*
\query{\text{Politician}} \Sigma^*$ instead of the more verbose $\Sigma^* (\Sigma^* \land
\query{\text{Politician}}) \Sigma^*$, and $(\query{\text{Sportsperson}} \texttt{", "})^*
\query{\text{Sportsperson}}$ instead of the second SemRE above.
\end{note}

\begin{note}[Role of the alphabet]
Throughout this paper, we will assume that $\Sigma = \{ a, b, c, \dots \}$ is finite. Given the
number of ASCII and Unicode code points, this assumption needs deeper consideration while actually
implementing these algorithms. Our prototype implementation operates over streams of 8-bit bytes
(using Rust's default UTF-8 encoding) and supports three types of character classes:
\begin{inparaenum}[(\itshape a\upshape)]
\item The wildcard $\Sigma$ (or ``\texttt{.}''), which is matched by any symbol from the alphabet,
  $a \in \Sigma$. In other words, the character class $\Sigma$ is simply notational shorthand and a
  minor optimization over the SemRE $a + b + c + \cdots$.
\item $\texttt{[}a\text{--}b\texttt{]}$ which matches all characters $c \in \Sigma$ such that
  $a \leq c \leq b$, and
\item $\texttt{[\^{}}a\text{--}b\texttt{]}$, which matches all characters $c \in \Sigma$ that do not
  match $\texttt{[}a\text{--}b\texttt{]}$.
\end{inparaenum}
By constructing an effective Boolean algebra over character classes, symbolic alphabets~%
\cite{Veanes:POPL2012, Saarikivi:TACAS2019} provide a more principled approach to these issues, and
is a good direction for future work.
\end{note}


\begin{example}[Credential leaks]
\label{ex:semre:secrets}
A common class of software vulnerabilities arises when developers embed passwords, SSH keys, and
other secret credentials in their source code, which subsequently gets accidentally publicized as
part of an open source software release. Several popular programs (such as Gitleaks~\cite{Gitleaks})
have been written to automatically discover hardcoded secrets in source code.
One approach to discover these forgotten credentials is to examine all string literals occurring in
the repository and flag those which look like secrets. While classical regular expressions are ideal
for recognizing string literals, an external oracle, such as an LLM, is better suited to determine
whether a specific string literal might be a password or SSH key. This approach can be summarized
using the following semantic regular expression:\footnote{After accounting for escape sequences, as
per \url{https://docs.oracle.com/javase/specs/jls/se22/html/jls-3.html\#jls-3.10.5}.}
%
%
\begin{alignat}{1}
  r_{\text{pass}} & = \texttt{"}
                      (\Sigma_s + \text{Esc})^* \land \query{\text{Password or SSH key}}
                      \texttt{"}, \text{ where} \\
  \Sigma_s        & = \Sigma \setminus \{ \texttt{"}, \text{BS} \} 
                      \text{ and} \nonumber \\
  \text{Esc}      & = \text{BS} \; \{ \texttt{b}, \texttt{t}, \texttt{n}, \texttt{f}, \texttt{r},
                                      \texttt{"}, \texttt{'}, \text{BS} \}, \nonumber
\end{alignat}
and where $\text{BS}$ is the backslash character.
\end{example}

\begin{assumption}[Deterministic oracles]
\label{assm:semre:llm}
We have implicitly assumed that $\oracle : Q \times \Sigma^* \to \Bool$ is a deterministic function.
On the other hand, the token sequences produced by LLMs are fundamentally non-deterministic. This
arises from multiple factors, including its temperature, sampling strategy, and even non-determinism
arising from floating point calculations performed on the GPU. We compensate for these issues by
lowering the temperature to 0 and caching results to previously submitted queries, thereby
forcefully determinizing the oracle.
\end{assumption}

\begin{example}[Non-existent file paths]
\label{ex:semre:files}
Source code also frequently contains hardcoded file paths, which progressively become obsolete as
bitrot sets in. A SemRE similar to the following might be used to automatically detect such paths:
\begin{alignat}{1}
  r_{\text{file}} & = (\Sigma_f^* \; \text{FS} \; 
                        (\Sigma_f^* + \text{FS})^+ +  
                       \Sigma_f^+ \; \text{FS})  \land
                      \query{\text{Non-existent file path}}, \text{ where} \\
  \Sigma_f        & = \{ \texttt{a}, \dots, \texttt{z},
                         \texttt{A}, \dots, \texttt{Z},
                         \texttt{0}, \dots, \texttt{9},
                         \texttt{-}, \text{PD}, \text{US} \}, \text{ and where} \nonumber
\end{alignat}
$\text{PD}$ and $\text{US}$ are the period and underscore characters, and $\text{FS}$ is the forward
slash character commonly used to separate directories in Unix file paths.
\end{example}

\begin{note}
In the above example, the oracle would be an external script that determines the (non-)existence of
specific file paths. Although many SemREs would use LLMs as the background oracle, $\oracle$, we
emphasize that the oracle may also be constructed using other sources of information, such as
databases, calculators, or other external tools.
\end{note}

\begin{example}[Identifier naming conventions]
\label{ex:semre:identifiers}
Software maintainers love to compile elaborate conventions for how to name variables, identifiers,
and other program elements,\footnote{For example:
\url{https://google.github.io/styleguide/javaguide.html\#s5-naming.}} and also love to complain
about programmers not following these conventions.
Ignoring keywords and non-ASCII characters, Java identifiers may be recognized using the regular
expression $\Sigma_l (\Sigma_l + \{ \texttt{0}, \dots, \texttt{9} \})^*$, where $\Sigma_l =
\{ \texttt{a}, \dots, \texttt{z}, \texttt{A}, \dots, \texttt{Z}, \texttt{\$}, \text{US} \}$.
With an appropriate oracle, they simply need to scan their project for identifiers of the form:
\begin{align}
  & r_{\text{id}}  = 
  (\Sigma_l \; (\Sigma_l + \{ \texttt{0}, \dots, \texttt{9} 
  \})^*)  \; \land 
            \query{\text{Inappropriately named Java identifier}}      
\end{align}
where the oracle presumably confirms that the name is meaningful, correctly uses camel case or snake
case formats, and is drawn from the appropriate grammatical categories (such as verbs or nouns).
\end{example}


\begin{example}[Pharmaceutical spam]
\label{ex:semre:viagra}
Spam email routinely contains advertisements for pharmaceutical products and weight loss
supplements. Assume access to an LLM which can reliably recognize the names of these products. One
might then use the following SemRE to flag the subject lines of these email messages:
\begin{alignat}{1}
  r_{\text{spam}, 1} & = \texttt{Subject:} \; \Sigma^* \; \queryp{\text{Medicine name}} \; \Sigma^*,
\end{alignat}
where $\queryp{q}$ is shorthand for $\Sigma^+ \land \query{q}$. One might additionally wish to
restrict their search to single whole words, so they might use the following SemRE instead:
\begin{alignat}{1}
  r_{\text{spam}, 2} & = \texttt{Subject:} \; \Sigma^* \;
                         \text{WS} \; \{ \texttt{a}, \dots, \texttt{z},
                                         \texttt{A}, \dots, \texttt{Z} \}^+ \land
                                      \query{\text{Medicine name}} \; \text{WS} \;
                         \Sigma^*,
\end{alignat}
where $\text{WS}$ is the whitespace character.
For a somewhat humorous example of classical regular expressions being used in a similar setting, we
refer the reader to Wikipedia's blacklist of forbidden article titles.%
\footnote{\url{https://en.wikipedia.org/wiki/MediaWiki:Titleblacklist}.}
\end{example}

\begin{example}[Domains, 1]
\label{ex:semre:whois}
Another common indicator of spam email is when the sender's domain no longer exists. This can be
checked using services such as DNS, Whois or PyFunceble~\cite{Pyfunceble}. One can then filter email
from senders that satisfy the pattern:
\begin{alignat}{1}
  r_{\text{edom}} & = \Sigma_e^+ \; \texttt{@} \;
                      (\Sigma_e^+ \; \text{PD} \; \Sigma_a^{\{ 1, 3 \}}) \land
                      \query{\text{Domain does not exist}}, \text{ where} \\
  \Sigma_e        & = \{ \texttt{a}, \dots, \texttt{z},
                         \texttt{A}, \dots, \texttt{Z},
                         \texttt{0}, \dots, \texttt{9},
                         \texttt{-}, \text{PD} \}, \text{ and} \nonumber \\
  \Sigma_a       & = \{ \texttt{a}, \dots, \texttt{z}, \texttt{A}, \dots, \texttt{Z} \}. \nonumber
\end{alignat}
Here the notation $r^{\{ i, j \}}$ is syntactic sugar for bounded repetition:
\[ r^{\{ i, j \}} = r^i + r^{i + 1} + \dots + r^j. \]
\end{example}

\begin{example}[Domains, 2]
\label{ex:semre:phishing}
A third characteristic of spam email is the presence of URLs from phishing domains. Security
researchers make repositories of these domains publicly available. See, for example,
\url{https://openphish.com/}. When scanning their inbox, one might therefore wish to run the SemRE:
\begin{alignat}{1}
  r_{\text{wdom}, 1} & = (\texttt{http} (\texttt{s}?) \texttt{://} + \texttt{www} \; \text{PD}) \;
                         (\Sigma_e^+ \; \text{PD} \; \Sigma_a^{\{ 1, 3 \}}) \land
                         \query{\text{Phishing domain}}.
\end{alignat}
We use the usual syntactic sugar $r?$ for $r + \strempty$.

In other situations, one might be interested in recent domains, say those registered after 2010. One
can once again use the Whois service to answer this question:
\begin{alignat}{1}
  r_{\text{wdom}, 2} & = (\texttt{http} (\texttt{s}?) \texttt{://} + \texttt{www} \; \text{PD}) \;
                         (\Sigma_e^+ \; \text{PD} \; \Sigma_a^{\{ 1, 3 \}}) \land
                         \query{\text{Domain registered after 2010}}.
\end{alignat}
\end{example}

\begin{example}[Foreign IP addresses.]
\label{ex:semre:ip}
While analyzing some packet capture data, a security researcher might be interested in mentions of
IP addresses that are outside the company's intranet. They might use their knowledge of the local
network topology to set up an oracle and look for matches for the following SemRE:
\begin{alignat}{1}
  r_{\text{ip}} & = ((\Sigma_d^{\{ 1, 3 \}} \; \text{PD})^3 \; \Sigma_d^{\{ 1, 3 \}}) \land
                    \query{\text{Foreign IP address}},
\end{alignat}
where $\Sigma_d = \{ \texttt{0}, \dots, \texttt{9} \}$.
\end{example}

%% file: src/alg.tex
\section{Matching Semantic Regular Expressions Using Query Graphs}
\label{sec:alg}

We will now present an algorithm to determine whether $w \in \interp{r}$ for a given SemRE $r$ and
input string $w$. Our procedure conceptually proceeds in two passes: It first uses an NFA-like
algorithm to identify syntactic structure in $w$. Notably, we do not discharge any non-trivial
oracle queries during this pass, but instead construct a data structure called a query graph which
encodes the Boolean structure among the results of these oracle queries. The second pass uses
dynamic programming to evaluate the query graph, discharge the necessary oracle queries, and produce
the final result.

\mclearpage
\input{src/alg-snfa.tex}          \mclearpage
\input{src/alg-qgraph-defn.tex}   \mclearpage
\input{src/alg-qgraph-build.tex}  \mclearpage
\input{src/alg-qgraph-eval.tex}   \mclearpage
\input{src/alg-analysis.tex}

%% file: src/alg-snfa.tex
\subsection{Converting SemREs into Semantic NFAs}
\label{sub:alg:snfa}

We start with a straightforward extension of classical NFAs: A \emph{semantic NFA} (SNFA) is a
5-tuple, $M = (S, \Delta, \lambda, s_0, s_f)$, where:
\begin{inparaenum}[(\itshape a\upshape)]
\item $S$ is a finite set of states,
\item $\Delta \subseteq S \times (\Sigma \union \{ \strempty \}) \times S$ is the transition
  relation,
\item $\lambda : S \to \{ \open, \close \} \times Q \union \{ \blank \}$ is the state labeling
  function,
\item $s_0, s_f \in S$ are the start and end states respectively, and such that:
\item for each sequence of transitions, $s_0 \to s_1 \to s_2 \to \dots \to s_f$, the corresponding
  sequence of state labels $\lambda(s_0) \lambda(s_1) \lambda(s_2) \cdots \lambda(s_f)$ forms a
  well-parenthesized string ($|Q|$ different types of parentheses, one for each query in $Q$).
\end{inparaenum}


\begin{remark}
The main difference with classical NFAs is the state labeling $\lambda(s)$. We use this to mark
substrings that subsequently need to be submitted to the oracle for validation. The semantics of
SemREs in Equation~\ref{eq:semre:semantics} imply that these substrings are well-nested.
Well-parenthesized structures will consequently play a central role in our algorithm, including in
the above definition of SNFAs and in the definition of query graphs in Section~%
\ref{sub:alg:qgraph-defn}.
\end{remark}

We will write $s \to^a s'$ to indicate the transition $(s, a, s') \in \Delta$. When the state label
$\lambda(s)$ is of the form $(\open, q)$ or $(\close, q)$, we instead write $\open(q)$ and
$\close(q)$ respectively.

\input{src/figs/alg/snfa.tex}

Given a SemRE $r$, we construct an SNFA $M_r$ by following a traditional Thompson-like approach. We
visualize this construction in Figure~\ref{fig:alg:snfa}, but postpone its formal description to
Appendix~\refapp{app:alg:snfa}.
To argue its correctness, we need to reason about (feasible) paths through the SNFA.
A path $\pi = s_1 \to^{w_1} s_2 \to^{w_2} s_3 \to^{w_3} \dots \to^{w_n} s_{n + 1}$, for $n \geq 0$,
consists of an initial state $s_1$, and a (possibly empty) sequence of transitions emerging from
$s_1$. The corresponding string $\str(\pi)$ is the sequence of characters, $\str(\pi) = w_1 w_2 w_3
\dots w_n$.
(Informally,) We say that the path $\pi$ is \emph{feasible} if the oracle accepts all its queries.
Formally, $\pi$ is feasible if any of the following conditions hold:
\begin{enumerate}
\item $\pi = s_1$ is the empty path where $\lambda(s_1) = \blank$, or
\item $\pi = s_1 \to s_2 \to \dots \to s_n \to s_{n + 1}$, where $\lambda(s_1) = \open(q)$ and
  $\lambda(s_{n + 1}) = \close(q)$ for some query $q$, the subpath $s_2 \to \dots \to s_n$ is
  feasible, and $\oracle(q, \str(\pi)) = \true$, or
\item $\pi = \underbrace{s_1 \to \dots \to s_k}_{\pi_1} \to \underbrace{s_{k + 1} \to \dots \to
  s_n}_{\pi_2}$, where both subpaths $\pi_1$ and $\pi_2$ are feasible.
\end{enumerate}
We say that a string $w$ is accepted by an SNFA $M$ if there a feasible $w$-labelled path $\pi = s_0
\to \dots \to s_f$.
The following theorem asserts the correctness of the constructed SNFA $M_r$:
\begin{theorem}
\label{thm:alg:snfa}
Pick a SemRE $r$ and let $M_r$ be the SNFA resulting from the construction of Figure~%
\ref{fig:alg:snfa}.
For each string $w$, $w \in \interp{r}$ iff $M_r$ accepts $w$.
\end{theorem}
\begin{proof}[Proof sketch]
The forward direction can be proved by induction on the derivation of $w \in \interp{r}$.
To prove the converse, we start by observing that whenever a combined path
\[ \pi = \underbrace{s_1 \to \dots \to s_n}_{\pi_1} \to
         \underbrace{s_{n + 1} \to \dots \to s_{n + m}}_{\pi_2} \]
is feasible and both subpaths $\pi_1$ and $\pi_2$ are well-parenthesized, then $\pi_1$ and $\pi_2$
are both themselves feasible paths. The final proposition may then be proved by induction on $r$,
and by case analysis on the presented proof of feasibility.
\end{proof}

%% file: src/figs/alg/snfa.tex
\begin{figure}
\begin{subfigure}{.25\textwidth}
\centering
\scalebox{\scalefactor}{
\begin{tikzpicture}
\node [state, initial]   (s0)               {$s_0$};
\node [state, accepting] (sf) [right=of s0] {$s_f$};
\end{tikzpicture}
}
\caption{$\bot$.}
\label{sfig:alg:snfa:bot}
\end{subfigure}
\hfill
%
\begin{subfigure}{.25\textwidth}
\centering
\scalebox{\scalefactor}{
\begin{tikzpicture}
\node [state, initial]   (s0)               {$s_0$};
\node [state, accepting] (sf) [right=of s0] {$s_f$};
\path [->] (s0) edge [above] node {$\strempty$} (sf);
\end{tikzpicture}
}
\caption{$\strempty$.}
\label{sfig:alg:snfa:strempty}
\end{subfigure}
\hfill
%
\begin{subfigure}{.25\textwidth}
\centering
\scalebox{\scalefactor}{
\begin{tikzpicture}
\node [state, initial]   (s0)               {$s_0$};
\node [state, accepting] (sf) [right=of s0] {$s_f$};
\path [->] (s0) edge [above] node {$a$} (sf);
\end{tikzpicture}
}
\caption{$a$.}
\label{sfig:alg:snfa:char}
\end{subfigure}

\vspace{2em}

\begin{subfigure}{.5\textwidth}
\centering
\scalebox{\scalefactor}{
\begin{tikzpicture}
\node [state, initial] (s0) {$s_0$};

\node [state] (s10) [above right=of s0] {$s_{10}$};
\node [state] (s1f) [right=of s10] {$s_{1f}$};

\node [state] (s20) [below right=of s0] {$s_{20}$};
\node [state] (s2f) [right=of s20] {$s_{2f}$};

\node[state, accepting] (sf) [below right=of s1f] {$s_f$};

\path [->] (s0) edge [above] node {$\strempty$} (s10);
\path [->] (s1f) edge [above] node {$\strempty$} (sf);
\path [->] (s0) edge [below] node {$\strempty$} (s20);
\path [->] (s2f) edge [below] node {$\strempty$} (sf);

\begin{scope}[on background layer]
  \node [fill=black!10, fit=(s10) (s1f)] {$r_1$};
\end{scope}

\begin{scope}[on background layer]
  \node [fill=black!10, fit=(s20) (s2f)] {$r_2$};
\end{scope}
\end{tikzpicture}
}
\caption{$r_1 + r_2$.}
\label{sfig:alg:snfa:union}
\end{subfigure}
\hfill
%
\begin{subfigure}{.45\textwidth}
\centering
\scalebox{\scalefactor}{
\begin{tikzpicture}
\node [state, initial] (s0) {$s_0$};
\node [state] (sr0) [right=of s0] {$s_{r0}$};
\node [state] (srf) [right=of sr0] {$s_{rf}$};
\node [state, accepting] (sf) [below=of s0] {$s_f$};

\path [->] (s0) edge [above] node {$\strempty$} (sr0);
\path [->] (s0) edge [left] node {$\strempty$} (sf);
\path [->, bend right=45] (srf) edge [above] node {$\strempty$} (s0);

\begin{scope}[on background layer]
  \node [fill=black!10, fit=(sr0) (srf)] {$r$};
\end{scope}
\end{tikzpicture}
}
\caption{$r^*$.}
\label{sfig:alg:snfa:kstar}
\end{subfigure}

\vspace{2em}

\begin{subfigure}{.5\textwidth}
\centering
\scalebox{\scalefactor}{
\begin{tikzpicture}
\node [state, initial] (s0) {$s_0$};

\node [state] (s10) [right=of s0] {$s_{10}$};
\node [state] (s1f) [right=of s10] {$s_{1f}$};

\node [state] (s20) [right=of s1f] {$s_{20}$};
\node [state] (s2f) [right=of s20] {$s_{2f}$};

\node[state, accepting] (sf) [right=of s2f] {$s_f$};

\path [->] (s0) edge [above] node {$\strempty$} (s10);
\path [->] (s1f) edge [above] node {$\strempty$} (s20);
\path [->] (s2f) edge [below] node {$\strempty$} (sf);

\begin{scope}[on background layer]
  \node [fill=black!10, fit=(s10) (s1f)] {$r_1$};
\end{scope}

\begin{scope}[on background layer]
  \node [fill=black!10, fit=(s20) (s2f)] {$r_2$};
\end{scope}
\end{tikzpicture}
}
\caption{$r_1 r_2$.}
\label{sfig:alg:snfa:concat}
\end{subfigure}
\hfill{}
\begin{subfigure}{.45\textwidth}
\centering
\scalebox{\scalefactor}{
\begin{tikzpicture}
\node [state, initial, label=below:{$\open(q)$}] (s0) {$s_0$};
\node [state] (sr0) [right=of s0] {$s_{r0}$};
\node [state] (srf) [right=of sr0] {$s_{rf}$};
\node [state, accepting, label=below:{$\close(q)$}] (sf) [right=of srf] {$s_f$};

\path [->] (s0) edge [above] node {$\strempty$} (sr0);
\path [->] (srf) edge [above] node {$\strempty$} (sf);

\begin{scope}[on background layer]
  \node [fill=black!10, fit=(sr0) (srf)] {$r$};
\end{scope}
\end{tikzpicture}
}
\caption{$r \land \query{q}$.}
\label{sfig:alg:snfa:query}
\end{subfigure}

\Description{Recursive construction of the semantic NFA $M_r$ given a SemRE $r$.}
\caption{Recursive construction of the semantic NFA $M_r$ given a SemRE $r$. The states $s_0$ and
  $s_f$ of $M_{r \land \query{q}}$ are respectively labelled with the open and close query markers
  for $q$. The formal construction may be found in Appendix~\refapp{app:alg}.}
\label{fig:alg:snfa}
\end{figure}

%% file: src/alg-qgraph-defn.tex
\subsection{The Query Graph Data Structure}
\label{sub:alg:qgraph-defn}

Given a SemRE $r$ and string $w$, our goal is to determine whether $w \in \interp{r}$. So far, our
approach of converting $r$ into a corresponding SNFA $M_r$ has been fairly standard. Per Theorem~%
\ref{thm:alg:snfa}, we just need to determine whether a feasible $w$-labelled path exists from the
initial to the final states of $M_r$. In the classical setting, we would do this by maintaining the
set of reachable states, and determining whether the accepting state $s_f$ was reachable. More
formally, we would iteratively define the set of reachable states $S_w$ after processing $w$ as:
\begin{alignat*}{1}
  S_\strempty & = \{ s_1 \in S \mid s_0 \to^{\strempty^*} s_1 \}, \text{ and} \\
  S_{wa}      & = \{ s_3 \in S \mid \exists s_1 \in S_w, s_2 \in S \text{ such that }
                                    s_1 \to^a s_2 \to^{\strempty^*} s_3 \},
\end{alignat*}
where $s \to^{\strempty^*} s'$ denotes the transitive closure of $\strempty$-transitions in the NFA.
When $r$ is a (classical) regular expression, $w \in \interp{r}$ iff $s_f \in S_w$.


\paragraph{The complication (Part 1).}

The classical algorithm fundamentally relies on the following Markovian property of NFA executions:
One only needs to know $S_w \subseteq S$ and $w'$ to determine whether the combined string $ww' \in
\interp{r}$. On the other hand, consider the SNFA in Figure~\ref{fig:alg:qgraph-defn:complication1},
which accepts strings of the form $r_{\text{pal}} = \Sigma^* a \query{\text{pal}}$, and assume the
query $\text{pal}$ accepts palindromic strings. Consider the prefixes:
\begin{alignat*}{2}
  w_1 & = babc \text{ and } & w_2 & = bacb.
\end{alignat*}
After processing these prefixes, the machine may be in any of the states:
\[ S_{w_1} = S_{w_2} = \{ s_0, s_2, s_f \}. \]
Now consider the common suffix,
\[ w_3 = cb. \]
Observe that the combined string $w_1 w_3 = babccb\in \interp{r_{\text{pal}}}$, while $w_2 w_3 =
bacbcb \notin \interp{r_{\text{pal}}}$. It follows that the SemRE membership testing algorithm needs
to maintain additional information beyond just the frontier $S_w$ of currently reachable states.

\begin{figure}
\centering
\scalebox{\scalefactor}{
\begin{tikzpicture}
  \node [state, initial] (s0) {$s_0$};
  \node [state, label=below:{$\open(\text{pal})$}, right=of s0] (s1) {$s_1$};
  \node [state, right=of s1] (s2) {$s_2$};
  \node [state, label=below:{$\close(\text{pal})$}, accepting, right=of s2] (sf) {$s_f$};

  \path [->] (s0) edge [loop above] node {$\Sigma$} (s0);
  \path [->] (s0) edge [above] node {$a$} (s1);
  \path [->] (s1) edge [above] node {$\strempty$} (s2);
  \path [->] (s2) edge [loop above] node {$\Sigma$} (s2);
  \path [->] (s2) edge [above] node {$\strempty$} (sf);
\end{tikzpicture}
}
\Description{SNFA which accepts strings of the form $\Sigma^* a \query{\text{pal}}$.}
\caption{SNFA which accepts strings of the form $\Sigma^* a \query{\text{pal}}$. Assume the query
  $\text{pal}$ recognizes palindromes. The machine nondeterministically finds an occurrence of the
  character $a$, and confirms that the subsequent suffix is a palindrome.}
\label{fig:alg:qgraph-defn:complication1}
\end{figure}


\paragraph{The complication (Part 2).}

Even for a single string, different paths leading to the same state may behave differently. For
example, consider the string $w_4 = babca$, so that
\begin{alignat*}{1}
  w_4 w_3 & =   ba\overbrace{bca\underbrace{cb}_{\cancel{\query{\text{pal}}}}}^{\query{\text{pal}}}
            \in \interp{r_{\text{pal}}}.
\end{alignat*}
The machine may nondeterministically take the $s_0
\to^a s_1$ transition either on the first or the second occurrence of the symbol $a$, resulting the
two paths $\pi_4$ and $\pi'_4$ in Figure~\ref{fig:alg:qgraph-defn:complication2}.
Both of these paths can then be extended using the path $\pi_3$ (corresponding to the suffix $w_3$).
Notice that the combined path $\pi_4 \pi_3$ is feasible, while $\pi'_4 \pi_3$ is not feasible, even
though both $\pi_4$ and $\pi'_4$ end at the same state $s_2$. Both paths $\pi_4$ and $\pi'_4$ have
an outstanding open query; the difference is the position in the string at which this query was
opened.
For each path, the SNFA evaluation algorithm therefore also needs to maintain the string indices at
which queries were opened and closed.

\begin{figure}
\centering
\scalebox{\scalefactor}{
\begin{tikzpicture}
  \node [state, initial] (s_e_0) {$s_0$};
  \node [state, right=of s_e_0] (s_b_0) {$s_0$};
  \node [state, right=of s_b_0, label=above:{$\open(\text{pal})$}] (s_ba_1) {$s_1$};
  \node [state, right=of s_ba_1] (s_ba_2) {$s_2$};
  \node [state, right=of s_ba_2] (s_bab_2) {$s_2$};
  \node [state, right=of s_bab_2] (s_babc_2) {$s_2$};

  \path [->] (s_e_0) edge [above] node {$b$} (s_b_0);
  \path [->] (s_b_0) edge [above] node {$a$} (s_ba_1);
  \path [->] (s_ba_1) edge [above] node {$\strempty$} (s_ba_2);
  \path [->] (s_ba_2) edge [above] node {$b$} (s_bab_2);
  \path [->] (s_bab_2) edge [above] node {$c$} (s_babc_2);

  \node [state, initial, below=of s_e_0] (t_e_0) {$s_0$};
  \node [state, right=of t_e_0] (t_b_0) {$s_0$};
  \node [state, right=of t_b_0] (t_ba_0) {$s_0$};
  \node [state, right=of t_ba_0] (t_bab_0) {$s_0$};
  \node [state, right=of t_bab_0] (t_babc_0) {$s_0$};
  \node [state, right=of t_babc_0, label=below:{$\open(\text{pal})$}] (t_babca_1) {$s_1$};

  \path [->] (t_e_0) edge [above] node {$b$} (t_b_0);
  \path [->] (t_b_0) edge [above] node {$a$} (t_ba_0);
  \path [->] (t_ba_0) edge [above] node {$b$} (t_bab_0);
  \path [->] (t_bab_0) edge [above] node {$c$} (t_babc_0);
  \path [->] (t_babc_0) edge [above] node {$a$} (t_babca_1);

  \node [state] (s_babca_2) at ($(s_babc_2)!0.5!(t_babca_1) + (2, 0)$) {$s_2$};
  \node [state, right=of s_babca_2] (s_babcac_2) {$s_2$};
  \node [state, right=of s_babcac_2] (s_babcacb_2) {$s_2$};
  \node [state, accepting, right=of s_babcacb_2, label=below:{$\close(\text{pal})$}]
    (s_babcacb_f) {$s_f$};

  \path [->] (s_babc_2) edge [above] node (pi4_a) {$a$} (s_babca_2);
  \path [->] (t_babca_1) edge [below] node {$\strempty$} (s_babca_2);

  \path [->] (s_babca_2) edge [above] node {$c$} (s_babcac_2);
  \path [->] (s_babcac_2) edge [above] node {$b$} (s_babcacb_2);
  \path [->] (s_babcacb_2) edge [above] node {$\strempty$} (s_babcacb_f);

  \node (pi4_bl) at ($(s_e_0) + (-0.5, 1.25)$) {};
  \node (pi4_br) at (pi4_bl -| pi4_a) {};
  \draw [decorate, decoration=brace] (pi4_bl) -- (pi4_br) node [midway, above] {$\pi_4$};

  \node (pip4_bl) at ($(t_e_0) + (-0.5, -1.25)$) {};
  \node (pip4_br) at (pip4_bl -| pi4_a) {};
  \draw [decorate, decoration={brace, mirror}] (pip4_bl) -- (pip4_br) node [midway, below] {$\pi'_4$};

  \node (pi3_bl) at (pi4_br) {};
  \node (pi3_br) at (pi3_bl -| s_babcacb_f) {};
  \draw [decorate, decoration=brace] (pi3_bl) -- (pi3_br) node [midway, above] {$\pi_3$};
\end{tikzpicture}
}
\Description{Figure illustrating that distinct $w$-labelled paths might not be equifeasible.}
\caption{Two prefix paths $\pi_4$ and $\pi'_4$ from the initial state $s_0$ to the intermediate
  state $s_2$. The SNFA in question is the one from Figure~\ref{fig:alg:qgraph-defn:complication1}.
  Both paths correspond to the same string $w_4 = babca$.
  The suffix path $\pi_3$ moves the machine from $s_2$ to the final state $s_f$ along the string
  $w_3 = cb$.
  Notice that the combined path $\pi_4 \pi_3$ is feasible (so the machine accepts $w_4 w_3 =
  babcacb$) but $\pi'_4 \pi_3$ is not. The SNFA evaluation algorithm therefore needs to track the
  indices at which queries were opened along each path through the automaton.}
\label{fig:alg:qgraph-defn:complication2}
\end{figure}


\paragraph{Our solution.}

Our solution is to unroll the SNFA while processing the input string. We maintain a compact
representation of all possible paths in a data structure called the query graph.
Fix a string $w = w_1 w_2 \dots w_n$, and recall that $Q$ is the set of all possible oracle queries.
Formally, a \emph{query graph}, $G = (V, E, \qgstart, \qgend, \qgidx, l)$, is a directed acyclic
graph (DAG) with vertices $V$ and edges $E \subseteq V \times V$,
with identified initial and final vertices, $\qgstart, \qgend \in V$, and
whose vertices are associated with string indices, $\qgidx : V \to \{ 1, 2, \dots, n + 1 \}$ and
labelled using $l : V \to \{ \open, \close \} \times Q \union \{ \blank \}$
such that for each path $p$ from $\qgstart$ to $\qgend$:
\begin{enumerate}
\item The index labels are monotonically increasing along $p$, and
\item the $\open$ and $\close$ marks and associated queries form a well-parenthesized string.
\end{enumerate}
The query graph is essentially an unrolling of the SNFA. Therefore, as before, we define its
semantics in terms of the feasibility of its paths. Let $p = v_1 \to v_2 \to \dots \to v_k$ be a
path through the query graph starting from a specified state $v_1$. We say that $p$ is feasible if
any of the following conditions hold:
\begin{enumerate}
\item $p = v_1$ is the empty path, where $l(v_1) = \blank(i_1)$ for some $i_1$, or
\item $p = v_1 \to v_2 \to \dots \to v_k \to v_{k + 1}$, where:
  \begin{enumerate}
  \item the subpath $v_2 \to \dots \to v_k$ is feasible, and
  \item $l(v_1) = \open(q)$ and $l(v_{k + 1}) = \close(q)$ for some query $q \in Q$, and
  \item $\oracle(q, w_i w_{i + 1} \cdots w_{j - 1}) = \true$, where $i = \qgidx(v_1)$ and $j =
    \qgidx(v_{k + 1})$, or
  \end{enumerate}
\item $p = \underbrace{v_1 \to \dots \to v_k}_{p_1} \to \underbrace{v_{k + 1} \to \dots \to
  v_n}_{p_2}$, where both subpaths $p_1$ and $p_2$ are feasible.
\end{enumerate}
We associate the entire query graph $G$ with a Boolean value $\interp{G} \in \Bool$, such that
$\interp{G} = \true$ iff there is a feasible path from $\qgstart$ to $\qgend$.

\input{src/figs/alg/qg-examples.tex}

We present some examples of query graphs in Figure~\ref{fig:alg-qgraph-defn:examples}. Graph $G_1$
corresponds to the string $w_4 w_3 = babcacb$ and the pattern $\Sigma^* a \query{\text{pal}}$ from
earlier in this section.
In general, notice that the answer to the membership problem, i.e., whether $w \in \interp{r}$,
depends on some Boolean combination of the results of no more than $O(|r| |w|^2)$ oracle queries. By
allowing a single $\close$ node to demarcate multiple open queries, Figure~%
\ref{sfig:alg-qgraph-defn:examples:pol} illustrates how this Boolean function might nevertheless be
encoded using a query graph with only $O(|r| |w|)$ vertices.
Of course, it is not yet clear that small query graphs can be constructed for arbitrary instances of
the membership testing problem: We devote Section~\ref{sub:alg:qgraph-build} to actually
constructing a compact query graph such that $\interp{G_r^w}$ is $\true$ iff $w \in \interp{r}$. It
is also not clear how $\interp{G}$ can be computed: We show how to solve this problem in
Section~\ref{sub:alg:qgraph-eval}.

%% file: src/figs/alg/qg-examples.tex
\begin{figure}
\begin{subfigure}{.2\textwidth}
\centering
\scalebox{\scalefactor}{
\begin{tikzpicture}
  \node [qgstart] (start) at (0, 0) {$\qgstart$};
  \node (o1) at (-1.2, 1) {$3 \qgat \open(\text{pal})$};
  \node (o2) at (1.2, 1) {$6 \qgat \open(\text{pal})$};
  \node (c8) at (0, 2) {$8 \qgat \close(\text{pal})$};
  \node [qgend] (end) at (0, 3) {$\qgend$};

  \path [->] (start) edge (o1);
  \path [->] (start) edge (o2);
  \path [->] (o1) edge (c8);
  \path [->] (o2) edge (c8);
  \path [->] (c8) edge (end);
\end{tikzpicture}
}
\caption{$G_1$.}
\label{sfig:alg-qgraph-defn:examples:w4w3-rpal}
\end{subfigure}
\hfill
\begin{subfigure}{.5\textwidth}
\centering
\scalebox{\scalefactor}{
\begin{tikzpicture}
  \node [qgstart] (start) at (0, 0) {$\qgstart$};

  \node [qgblank] (a1) at (0, 0.75) {};
  \node [qgblank] (a2) at (2, 0.75) {};
  \node [qgblank] (a3) at (4, 0.75) {};
  \node [qgblank] (a4) at (6, 0.75) {};
  \node [qgblank] (a5) at (8, 0.75) {};

  \path [->] (start) edge (a1);
  \path [->] (a1) edge (a2);
  \path [->] (a2) edge (a3);
  \path [->] (a3) edge (a4);
  \path [->] (a4) edge (a5);

  \node (b1) at (0, 1.5) {$1 \qgat \open(q)$};
  \node (b2) at (2, 1.5) {$2 \qgat \open(q)$};
  \node (b3) at (4, 1.5) {$3 \qgat \open(q)$};
  \node (b4) at (6, 1.5) {$4 \qgat \open(q)$};
  \node (b5) at (8, 1.5) {$5 \qgat \open(q)$};

  \path [->] (a1) edge (b1);
  \path [->] (a2) edge (b2);
  \path [->] (a3) edge (b3);
  \path [->] (a4) edge (b4);
  \path [->] (a5) edge (b5);

  \node [qgblank] (c1) at (0, 2.25) {};
  \node [qgblank] (c2) at (2, 2.25) {};
  \node [qgblank] (c3) at (4, 2.25) {};
  \node [qgblank] (c4) at (6, 2.25) {};
  \node [qgblank] (c5) at (8, 2.25) {};

  \path [->] (b1) edge (c1);
  \path [->] (b2) edge (c2);
  \path [->] (b3) edge (c3);
  \path [->] (b4) edge (c4);
  \path [->] (b5) edge (c5);

  \path [->] (c1) edge (c2);
  \path [->] (c2) edge (c3);
  \path [->] (c3) edge (c4);
  \path [->] (c4) edge (c5);

  \node (d1) at (0, 3) {$1 \qgat \close(q)$};
  \node (d2) at (2, 3) {$2 \qgat \close(q)$};
  \node (d3) at (4, 3) {$3 \qgat \close(q)$};
  \node (d4) at (6, 3) {$4 \qgat \close(q)$};
  \node (d5) at (8, 3) {$5 \qgat \close(q)$};

  \path [->] (c1) edge (d1);
  \path [->] (c2) edge (d2);
  \path [->] (c3) edge (d3);
  \path [->] (c4) edge (d4);
  \path [->] (c5) edge (d5);

  \node [qgblank] (e1) at (0, 3.75) {};
  \node [qgblank] (e2) at (2, 3.75) {};
  \node [qgblank] (e3) at (4, 3.75) {};
  \node [qgblank] (e4) at (6, 3.75) {};
  \node [qgblank] (e5) at (8, 3.75) {};

  \path [->] (d1) edge (e1);
  \path [->] (d2) edge (e2);
  \path [->] (d3) edge (e3);
  \path [->] (d4) edge (e4);
  \path [->] (d5) edge (e5);

  \path [->] (e1) edge (e2);
  \path [->] (e2) edge (e3);
  \path [->] (e3) edge (e4);
  \path [->] (e4) edge (e5);

  \node [qgend] (end) at (8, 4.5) {$\qgend$};

  \path [->] (e5) edge (end);
\end{tikzpicture}
}
\caption{$G_2$.}
\label{sfig:alg-qgraph-defn:examples:pol}
\end{subfigure}
\hfill
\begin{subfigure}{.2\textwidth}
\centering
\scalebox{\scalefactor}{
\begin{tikzpicture}
  \node [qgstart] (start) at (0, 0) {$\qgstart$};
  \node (a1) at (0, 1) {$3 \qgat \open(q)$};
  \node (b1) at (-1, 2) {$5 \qgat \open(q')$};
  \node (b2) at (1, 2) {$7 \qgat \open(q')$};
  \node (c1) at (0, 3) {$8 \qgat \close(q')$};
  \node (d1) at (0, 4) {$8 \qgat \close(q)$};
  \node [qgend] (end) at (0, 5) {$\qgend$};

  \path [->] (start) edge (a1);
  \path [->] (a1) edge (b1);
  \path [->] (a1) edge (b2);
  \path [->] (b1) edge (c1);
  \path [->] (b2) edge (c1);
  \path [->] (c1) edge (d1);
  \path [->] (d1) edge (end);
\end{tikzpicture}
}
\caption{$G_3$.}
\label{sfig:alg-qgraph-defn:examples:nest}
\end{subfigure}
\Description{Examples of query graphs.}
\caption{Examples of query graphs.
  \ref{sfig:alg-qgraph-defn:examples:w4w3-rpal}: Corresponding to possible parse trees for the
  string $w_4 w_3 = babcacb$ according to the SemRE $r_{\text{pal}} = \Sigma^* a
  \query{\text{pal}}$.
  We indicate the labels on each node $v$ by writing $\qgidx(v) \qgat l(v)$.
  The path through the $\open$ node on the left is feasible if $\oracle(\text{pal}, bcacb) = \true$
  and the path on the right is feasible if $\oracle(\text{pal}, cb) = \true$. $\interp{G_1} = \true$
  if either of these paths is feasible.
  \ref{sfig:alg-qgraph-defn:examples:pol}: Corresponding to the string $w = w_1 w_2 w_3 w_4$
  according to the pattern $\Sigma^* \query{q} \Sigma^*$. The unlabelled vertices are all marked
  $\blank$. Observe how each $\open$ node may be delimited by any subsequently reachable matching
  $\close$ node. Our construction in Section~\ref{sub:alg:qgraph-build} exploits similar sharing to
  produce a query graph with only $O(|r| |w|)$ vertices.
  \ref{sfig:alg-qgraph-defn:examples:nest}: Query graph with ``\emph{nested}'' queries.
  Corresponding to the string $w = babcbc$ and the SemRE $r_{\text{nest}} = \Sigma^* a (\Sigma^* b
  \query{q'}) \land \query{q}$. $\interp{G_3} = \true$ iff the Boolean formula $\oracle(q, cbcbc)
  \land (\oracle(q', cbc) \lor \oracle(q', c))$ evaluates to $\true$.}
\label{fig:alg-qgraph-defn:examples}
\end{figure}

%% file: src/alg-qgraph-build.tex
\subsection{Constructing the Query Graph}
\label{sub:alg:qgraph-build}

Given an SNFA $M$ and input string $w \in \Sigma^*$, our goal in this section is to construct a
query graph $G_M^w$ such that $\interp{G_M^w} = \true$ iff $M$ accepts $w$. We hope to construct
$G_M^w$ by making one left-to-right pass over the string. Furthermore, at a high level, we hope to
iteratively maintain a ``\emph{frontier}'' of vertices in $G_M^w$ corresponding to the reachable
states of $M$. After processing each subsequent character $w_i$, we create a new ``\emph{layer}'' of
vertices in the growing query graph, and add edges between these last two layers according to the
$w_i$-labelled transitions of $M$.

As an example, consider the SNFA $M_{q^*}$ shown in Figure~\ref{sfig:alg:qgraph-build:qstar:snfa},
which accepts strings of the form $(\Sigma^* \land \query{q})^*$. Let $w = abc$. Depending on how\
$w$ is partitioned, it is accepted in one of four different cases: when
\begin{equation}
\left.
\begin{aligned}
  \oracle(w, abc)                                       & = \true, \text{ or}      \\
  \oracle(w, a) \land \oracle(w, bc)                    & = \true, \text{ or}      \\
  \oracle(w, ab) \land \oracle(w, c)                    & = \true, \text{ or when} \\
  \oracle(w, a) \land \oracle(w, b) \land \oracle(w, c) & = \true.
\end{aligned}
\qquad \right\}
\label{eq:alg:graph-build:motiv-ex}
\end{equation}
Each of these alternatives corresponds to a different path through the SNFA. Observe that these
alternatives differ only in the sequence of $\strempty$-transitions exercised after processing each
character, which result in the query $q$ being closed (and possibly immediately reopened) at
different locations in the input string. The query graph $G_{q^*}^{abc}$ in Figure~%
\ref{sfig:alg:qgraph-build:qstar:qg} collectively summarizes the effect of these alternatives: The
edges labelled $a$, $b$ and $c$ in red indicate new layers of the query graph that are added after
each subsequent character is processed. These are straightforward to add. The remaining edges all
correspond to $\strempty$-transitions in the SNFA. The technical meat of this section is therefore
in designing an \emph{inter-character gadget} which summarizes the effect of $\strempty$-%
transitions.

\input{src/figs/alg/qgraph-build-qstar.tex}


\mclearpage
\input{src/alg-qgraph-build-prelims.tex} \mclearpage
\input{src/alg-qgraph-build-gadget.tex}  \mclearpage
\input{src/alg-qgraph-build-final.tex}

%% file: src/figs/alg/qgraph-build-qstar.tex
\begin{figure}
\begin{subfigure}{0.18\textwidth}
\centering
\scalebox{\scalefactor}{
\begin{tikzpicture}
  \node [state, initial, accepting] (s0) at (0, 0) {$s_0$};
  \node [state, label=right:{$\open(q)$}] (s1) at (0, -1.5) {$s_1$};
  \node [state] (s2) at (0, -3) {$s_2$};
  \node [state, label=right:{$\close(q)$}] (s3) at (0, -4.5) {$s_3$};

  \path [->] (s0) edge [right] node {$\strempty$} (s1);
  \path [->] (s1) edge [right] node {$\strempty$} (s2);
  \path [->, loop right] (s2) edge [right] node {$\Sigma$} (s2);
  \path [->] (s2) edge [right] node {$\strempty$} (s3);
  \path [->, bend left] (s3) edge [left] node {$\strempty$} (s0);
\end{tikzpicture}
}
\caption{}
\label{sfig:alg:qgraph-build:qstar:snfa}
\end{subfigure}
\hfill
\begin{subfigure}{0.8\textwidth}
\centering
\scalebox{\scalefactor}{
\begin{tikzpicture}
  \node [qgstart] (v00) at (0, 0) {$\qgstart$};
  \node (v01) at (0, -1.5) {$1 \qgat \open(q)$};
  \node [qgblank] (v02) at (0, -3) {};

  \node [qgblank] (v12) at (2.5, -3) {};
  \node (v13) at (2.5, -4.5) {$2 \qgat \close(q)$};

  \node [qgblank] (v20) at (5, 0) {};
  \node (v21) at (5, -1.5) {$2 \qgat \open(q)$};
  \node [qgblank] (v22) at (5, -3) {};

  \node [qgblank] (v32) at (7.5, -3) {};
  \node (v33) at (7.5, -4.5) {$3 \qgat \close(q)$};

  \node [qgblank] (v40) at (10, 0) {};
  \node (v41) at (10, -1.5) {$3 \qgat \open(q)$};
  \node [qgblank] (v42) at (10, -3) {};

  \node [qgblank] (v52) at (12.5, -3) {};
  \node (v53) at (12.5, -4.5) {$4 \qgat \close(q)$};

  \node [qgend] (v60) at (15, 0) {$\qgend$};

  \path [->] (v00) edge (v01);
  \path [->] (v01) edge (v02);
  \path [->, thick, red] (v02) edge [above] node {$a$} (v12);
  \path [->] (v12) edge (v13);
  \path [->] (v12) edge (v22);
  \path [->, bend left=10] (v13) edge (v20);
  \path [->] (v20) edge (v21);
  \path [->] (v21) edge (v22);
  \path [->, thick, red] (v22) edge [above] node {$b$} (v32);
  \path [->] (v32) edge (v33);
  \path [->] (v32) edge (v42);
  \path [->, bend left=10] (v33) edge (v40);
  \path [->] (v40) edge (v41);
  \path [->] (v41) edge (v42);
  \path [->, thick, red] (v42) edge [above] node {$c$} (v52);
  \path [->] (v52) edge (v53);
  \path [->, bend left=10] (v53) edge (v60);
\end{tikzpicture}
}
\caption{}
\label{sfig:alg:qgraph-build:qstar:qg}
\end{subfigure}

\Description{Figure illustrating the importance of $\strempty$-transitions in building the query
  graph.}
\caption{\ref{sfig:alg:qgraph-build:qstar:snfa}: The SNFA $M_{q^*}$ accepts strings of the form
  $(\Sigma^* \land \query{q})^*$.
  Let $w = abc$. The query graph $G_{q^*}^{abc}$ in Figure~\ref{sfig:alg:qgraph-build:qstar:qg}
  describes the four possible ways in Equation~\ref{eq:alg:graph-build:motiv-ex} by which $M_{q^*}$
  might accept $w$.
  Note that edges in the query graph are unlabelled: The red coloured labels $a$, $b$ and $c$ are
  there simply to hint to the reader that these edges can be thought of as originating from the
  corresponding $s_2 \to s_2$ transition of the SNFA.}
\label{fig:alg:qgraph-build:qstar}
\end{figure}

%% file: src/alg-qgraph-build-prelims.tex
\subsubsection{Preliminaries: Tentatively Feasible Subpaths, Query Contexts, and
  $\strempty$-Feasibility.}
\label{ssub:alg:qgraph-build:prelims}

Before discussing the gadget, we note that although every complete $s_0 \to^* s_f$ path is
well-parenthesized, its subpaths might have unmatched $\close$ and $\open$ states. We start by
investigating the structure of these subpaths.

Let $\pi = s_1 \to s_2 \to \dots \to s_k$ be a path through the SNFA starting at a specified state
$s_1$. We say that $\pi$ is \emph{tentatively well-parenthesized} if the sequence of state labels
$\lambda(\pi) = \lambda(s_1) \lambda(s_2) \dots \lambda(s_k)$ is a substring of some
well-parenthesized string. For example, the path in Figure~\ref{sfig:alg:qgraph-build:tentative:wp}
is tentatively well-parenthesized but the one in Figure~\ref{sfig:alg:qgraph-build:tentative:mp} is
not.

\input{src/figs/alg/tentative.tex}

Given a tentatively well-parenthesized path $\pi$, let $\sigma_c$ be the list of its unbalanced
close nodes, and let $\sigma_o$ be the list of its unbalanced open nodes. For example, for the path
$\pi_1$ in Figure~\ref{sfig:alg:qgraph-build:tentative:wp}, $\sigma_{c, 1} = [ \close(q_1) ]$ and
$\sigma_{o, 1} = [ \open(q_3) ]$. Given a tentatively well-parenthesized path $\pi$, we refer to the
pair $(\sigma_c, \sigma_o)$ as its \emph{query context}, which we will denote by $\qcon(\pi) =
(\sigma_c, \sigma_o)$.

We then note, without elaboration, that it is possible to compute the query context $\qcon(\pi)$ of
a combined path:
\[ \pi = \underbrace{s_1 \to^* s_2}_{\pi_1} \to \underbrace{s_3 \to^* s_4}_{\pi_2}, \]
merely by consulting $\qcon(\pi_1)$ and $\qcon(\pi_2)$.

Now, recall that all paths from the start state $s_0$ to the end state $s_f$ in the SNFA are
well-parenthesized. A natural consequence is that all paths through $M$ (regardless of origin and
destination) are also tentatively well-parenthesized. Furthermore, for each state $s$, and for each
pair of paths $\pi = s_0 \to^* s$ and $\pi' = s_0 \to^* s$, $\qcon(\pi) = \qcon(\pi')$. Because all
paths to a given state $s$ share the same query context, we can now speak of the query context
$\qcon(s)$ of individual states $s$, without identifying any particular path $\pi = s_0 \to^* s$.

\begin{assumption}
\label{assm:alg:qgraph-build:snfa1}
We assume that every state $s$ in $M$ is both reachable from the start state, $s_0 \to^* s$, and can
itself reach the final state, $s \to^* s_f$. The SNFA $M_r$ constructed in Section~%
\ref{sub:alg:snfa} automatically satisfies this property if $r$ does not contain any occurrence of
$\bot$. Such sub-expressions can be eliminated by the application of simple rewrite rules and using
the observation that $w \notin \interp{\bot}$.
\end{assumption}


We say that a tentatively well-parenthesized path $\pi$ through the SNFA is \emph{tentatively
feasible} if the oracle accepts all of its balanced queries. Similar ideas of tentative
well-parenthesization, query contexts, and tentative feasibility can also be developed for query
graphs.

Finally, we note that the feasibility of well-parenthesized $\strempty$-labelled paths depends on
responses from the oracle for queries of the form $\oracle(q, \strempty)$. Before constructing the
gadget, we therefore begin by determining all pairs of states, $(s_1, s_2)$, such that there exists
an $\strempty$-labelled feasible path from $s_1$ to $s_2$. In this case, we will write $s_1 \efp
s_2$. The set of all such pairs can be calculated using a graph search algorithm in $O(|r|^2)$
time. The details of this algorithm can be found in Appendix~\refapp{app:alg:qgraph-build:prelims}.

%% file: src/figs/alg/tentative.tex
\begin{figure}
\begin{subfigure}{0.45\textwidth}
\centering
\scalebox{\scalefactor}{
\begin{tikzpicture} 
  \node [state, label=below:{$\close(q_1)$}] (s1) at (0, 0) {$s_1$};
  \node [state, label=below:{$\open(q_2)$}] (s2) at (2, 0) {$s_2$};
  \node [state, label=below:{$\close(q_2)$}] (s3) at (4, 0) {$s_3$};
  \node [state, label=below:{$\open(q_3)$}] (s4) at (6, 0) {$s_4$};

  \path [->] (s1) edge (s2);
  \path [->] (s2) edge (s3);
  \path [->] (s3) edge (s4);
\end{tikzpicture}
}
\caption{}
\label{sfig:alg:qgraph-build:tentative:wp}
\end{subfigure}
\hfill
\begin{subfigure}{0.45\textwidth}
\centering
\scalebox{\scalefactor}{
\begin{tikzpicture} 
  \node [state, label=below:{$\close(q_1)$}] (s1) at (0, 0) {$s_1$};
  \node [state, label=below:{$\open(q_3)$}] (s2) at (2, 0) {$s_2$};
  \node [state, label=below:{$\close(q_2)$}] (s3) at (4, 0) {$s_3$};
  \node [state, label=below:{$\open(q_3)$}] (s4) at (6, 0) {$s_4$};

  \path [->] (s1) edge (s2);
  \path [->] (s2) edge (s3);
  \path [->] (s3) edge (s4);
\end{tikzpicture}
}
\caption{}
\label{sfig:alg:qgraph-build:tentative:mp}
\end{subfigure}

\Description{Examples and non-examples of tentatively well-parenthesized paths.}
\caption{(Non-)Examples of tentatively well-parenthesized paths. Assume $q_3 \neq q_2$. The path in
  Figure~\ref{sfig:alg:qgraph-build:tentative:wp} is tentatively well-parenthesized while the path
  in Figure~\ref{sfig:alg:qgraph-build:tentative:mp} is not.}
\label{sfig:alg:qgraph-build:tentative}
\end{figure}

%% file: src/alg-qgraph-build-gadget.tex
\subsubsection{The Inter-Character Gadget}
\label{ssub:alg:qgraph-build:gadget}

The purpose of the gadget is to summarize $\strempty$-labelled paths through the SNFA. As an
example, consider the SNFA in Figure~\ref{sfig:alg:qgraph-build:gadget:snfa}, and say that we are
currently in state $s_5$. Before processing the next character, the machine might choose to either
stay in the same state, or close and immediately reopen one or both of the queries $q_1$ and $q_2$.
Because of the well-parenthesized property of query opens and closes, $q_1$ must be closed
\emph{before} $q_2$ is closed, and can only be reopened \emph{after} $q_2$ is reopened. This
immediately suggests a three layer structure, where queries are sequentially closed in Layer~1,
(re-)opened in Layer~2, and where the last layer performs any remaining $\strempty$-transitions that
do not affect query context. We visualize the resulting gadget in Figure~%
\ref{sfig:alg:qgraph-build:gadget:qgraph}, and highlight the possible taken from state $s_5$ in
Figure~\ref{sfig:alg:qgraph-build:gadget:5e}.

\input{src/figs/alg/gadget.tex}

\input{src/figs/alg/epsilon-path.tex}


We formally define the gadget $\Gad$ as the triple:
{\allowdisplaybreaks
\begin{alignat}{1}
  \Gad      & = (V_\Gad, E_\Gad, l_\Gad), \text{ where} \label{eq:alg:qgraph-build:gadget} \\
  V_\Gad    & = S \times \{ 1, 2, 3 \}, \nonumber \\
  E_\Gad    & = E_{\Gad, 11} \union E_{\Gad, 12} \union
                E_{\Gad, 22} \union E_{\Gad, 23}, \text{ and} \nonumber \\
  l_\Gad(v) & = \begin{cases}
                \close(q) & \text{if } v = (s, 1) \text{ and } \lambda(s) = \close(q), \\
                \open(q)  & \text{if } v = (s, 2) \text{ and } \lambda(s) = \open(q), \text{ and} \\
                \blank    & \text{otherwise}.
                \end{cases} \nonumber
\end{alignat}
}
Here, $E_{\Gad, 11}$, $E_{\Gad, 12}$, $E_{\Gad, 22}$, and $E_{\Gad, 23}$ refer to edges between the
respective layers. Informally, edges in $E_{\Gad, 11}$ result in the progressive closing of
currently open queries and $E_{\Gad, 22}$ progressively (re-)opens queries. The direct connections
in $E_{\Gad, 12}$ provides shortcut connections to bypass these reopening actions. The last set of
edges, $E_{\Gad, 23}$ summarizes $\strempty$-labelled feasible paths that do not visibly change the
query context. These edges are formally defined as follows:
{\allowdisplaybreaks
\begin{alignat*}{1}
  E_{\Gad, 11} & = \left\{ (s, 1) \to (s''', 1) \middle|
                           \begin{array}{l}
                           s, s', s'', s''' \in S, q \in Q \text{ such that} \\
                           s' \efp s'' \text{ and }
                           s \to^\strempty s'
                           \text{ and }
                           s'' \to^\strempty s''' \text{ and }
                           \lambda(s''') = \close(q)
                           \end{array} \right\} \union \\
                    & \quad \left\{ (s, 1) \to (s''', 1) \middle|
                           \begin{array}{l}
                           s, s''' \in S, q \in Q \text{ such that} \\
                           s \to^\strempty s''
                           \text{ and }
                           \lambda(s''') = \close(q)
                           \end{array} \right\}, \\
  E_{\Gad, 12} & = \{ (s, 1) \to (s, 2) \mid s \in S \}, \\
  E_{\Gad, 22} & = \left\{ (s, 2) \to (s''', 2) \middle|
                           \begin{array}{l}
                           s, s', s'', s''' \in S, q \in Q \text{ such that} \\
                           s' \efp s'' \text{ and }
                           s \to^\strempty s'
                           \text{ and }
                           s'' \to^\strempty s''' \text{ and }
                           \lambda(s''') = \open(q)
                           \end{array} \right\} \union \\
                    & \quad \left\{ (s, 2) \to (s''', 2) \middle|
                           \begin{array}{l}
                           s, s''' \in S, q \in Q \text{ such that} \\
                           s \to^\strempty s''
                           \text{ and }
                           \lambda(s''') = \open(q)
                           \end{array} \right\}, \\
  E_{\Gad, 23} & = \{ (s, 2) \to (s', 3) \mid s, s', s'' \in S 
                      \text{ such that } s \efp s'' \text{ and } s'' \to^\strempty s'
                      \text{ or } s = s'\}.
\end{alignat*}
}
Our central claim about this construction is that there is a correspondence between
$\strempty$-labeled tentatively feasible paths in the SNFA and paths in the gadget. We formalize and
prove this claim in Lemma~\refapp{lem:qgraph-build:gadget:pc} of Appendix~%
\refapp{app:alg:qgraph-build:gadget}. We schematically visualize the breakdown of tentatively
feasible $\strempty$ paths across the three layers of the gadget in Figure~%
\ref{fig:alg:qgraph-build:gadget:eps}.

%% file: src/figs/alg/gadget.tex
\begin{figure}
\begin{subfigure}{0.18\textwidth}
\centering
\scalebox{\scalefactor}{
\begin{tikzpicture} 
  \node [state, initial] (s1) at (0, 0) {$s_1$};
  \node [state] (s2) at (0, -1.5) {$s_2$};
  \node [state, label=right:{$\open(q_2)$}] (s3) at (0, -3) {$s_3$};
  \node [state, label=right:{$\open(q_1)$}] (s4) at (0, -4.5) {$s_4$};
  \node [state] (s5) at (0, -6) {$s_5$};
  \node [state, label=right:{$\close(q_1)$}] (s6) at (0, -7.5) {$s_6$};
  \node [state, label=right:{$\close(q_2)$}] (s7) at (0, -9) {$s_7$};
  \node [state, accepting] (s8) at (0, -10.5) {$s_8$};

  \path [->] (s1) edge (s2);
  \path [->] (s2) edge (s3);
  \path [->, loop right] (s5) edge [right] node {$\Sigma$} (s5);
  \path [->] (s3) edge (s4);
  \path [->] (s4) edge (s5);
  \path [->] (s5) edge (s6);
  \path [->] (s6) edge (s7);
  \path [->, bend right] (s2) edge (s8);
  \path [->, bend left = 20] (s7) edge (s2);
  \path [->, bend left] (s6) edge (s4);
\end{tikzpicture}
}
\caption{}
\label{sfig:alg:qgraph-build:gadget:snfa}
\end{subfigure}
\hfill
\begin{subfigure}{.38\textwidth}
\centering
\scalebox{\scalefactor}{
\begin{tikzpicture}
  \node (q10) at (0, 1.5) {Layer 1};
  \node (q20) at (3, 1.5) {Layer 2};
  \node (q30) at (6, 1.5) {Layer 3};
  
  \node [qgblank] (q11) at (0, 0) {};
  \node [qgblank] (q12) at (0, -1.5) {};
  \node [qgblank] (q13) at (0, -3) {};
  \node [qgblank] (q14) at (0, -4.5) {};
  \node [qgblank] (q15) at (0, -6) {};
  \node (q16) at (0, -7.5) {$\close(q_1)$};
  \node (q17) at (0, -9) {$\close(q_2)$};
  \node [qgblank] (q18) at (0, -10.5) {};

  \node [qgblank] (q21) at (3, 0) {};
  \node [qgblank] (q22) at (3, -1.5) {};
  \node (q23) at (3, -3) {$\open(q_2)$};
  \node (q24) at (3, -4.5) {$\open(q_1)$};
  \node [qgblank] (q25) at (3, -6) {};
  \node [qgblank] (q26) at (3, -7.5) {};
  \node [qgblank] (q27) at (3, -9) {};
  \node [qgblank] (q28) at (3, -10.5) {};
  
  \node [qgblank] (q31) at (6, 0) {};
  \node [qgblank] (q32) at (6, -1.5) {};
  \node [qgblank] (q33) at (6, -3) {};
  \node [qgblank] (q34) at (6, -4.5) {};
  \node [qgblank] (q35) at (6, -6) {};
  \node [qgblank] (q36) at (6, -7.5) {};
  \node [qgblank] (q37) at (6, -9) {};
  \node [qgblank] (q38) at (6, -10.5) {};
  
  \path [->, bend right] (q13) edge (q17);
  \path [->, bend right] (q14) edge (q16);
  \path [->, bend right] (q15) edge (q16);
  \path [->, bend right] (q16) edge (q17);

  \path [->] (q21) edge (q32);
  \path [->] (q21) edge (q38);
  \path [->] (q21) edge (q32);
  \path [->, in = 180, out =0] (q22) edge (q38);
  \path [->, in = 180, out =0] (q22) edge (q37);
  \path [->, bend left = 10] (q23) edge (q36);
  \path [->] (q24) edge (q35);
  \path [->, bend right = 15] (q27) edge (q32);
  \path [->] (q27) edge (q38);

  \path [->, bend right] (q21) edge (q23);
  \path [->] (q22) edge (q23);
  \path [->] (q23) edge (q24);
  \path [->, bend left] (q26) edge (q24);
  \path [->, bend left] (q27) edge (q23);

  \path [->] (q11) edge (q21);
  \path [->] (q12) edge (q22);
  \path [->] (q13) edge (q23);
  \path [->] (q14) edge (q24);
  \path [->] (q15) edge (q25);
  \path [->] (q16) edge (q26);
  \path [->] (q17) edge (q27);
  \path [->] (q18) edge (q28);
  
  \path [->] (q21) edge (q31);
  \path [->] (q22) edge (q32);
  \path [->] (q23) edge (q33);
  \path [->] (q24) edge (q34);
  \path [->] (q25) edge (q35);
  \path [->] (q26) edge (q36);
  \path [->] (q27) edge (q37);
  \path [->] (q28) edge (q38);
\end{tikzpicture}
}
\caption{}
\label{sfig:alg:qgraph-build:gadget:qgraph}
\end{subfigure}
\hfill
\begin{subfigure}{.38\textwidth}
\centering
\scalebox{\scalefactor}{
\begin{tikzpicture}
  \node (q10) at (0, 1.5) {Layer 1};
  \node (q20) at (3, 1.5) {Layer 2};
  \node (q30) at (6, 1.5) {Layer 3};

  \node [qgblank] (q15) at (0, -6) {};
  \node (q16) at (0, -7.5) {$\close(q_1)$};
  \node (q17) at (0, -9) {$\close(q_2)$};
  
  \node (q23) at (3, -3) {$\open(q_2)$};
  \node (q24) at (3, -4.5) {$\open(q_1)$};
  \node [qgblank] (q25) at (3, -6) {};
  \node [qgblank] (q26) at (3, -7.5) {};
  \node [qgblank] (q27) at (3, -9) {};

  \node [qgblank] (q32) at (6, -1.5) {};
  \node [qgblank] (q33) at (6, -3) {};
  \node [qgblank] (q34) at (6, -4.5) {};
  \node [qgblank] (q35) at (6, -6) {};
  \node [qgblank] (q36) at (6, -7.5) {};
  \node [qgblank] (q37) at (6, -9) {};
  \node [qgblank] (q38) at (6, -10.5) {};

  \path [->, bend right] (q15) edge (q16);
  \path [->, bend right] (q16) edge (q17);

  \path [->] (q15) edge (q25);
  \path [->] (q16) edge (q26);
  \path [->] (q17) edge (q27);

  \path [->] (q23) edge (q24);
  \path [->, bend left] (q26) edge (q24);
  \path [->, bend left] (q27) edge (q23);
  \path [->] (q27) edge (q38);

  \path [->] (q23) edge (q33);
  \path [->] (q24) edge (q34);
  \path [->] (q25) edge (q35);
  \path [->] (q26) edge (q36);
  \path [->] (q27) edge (q37);

  \path [->, bend left = 10] (q23) edge (q36);
  \path [->] (q24) edge (q35);
  \path [->, bend right = 15] (q27) edge (q32);

\end{tikzpicture}
}
\caption{}
\label{sfig:alg:qgraph-build:gadget:5e}
\end{subfigure}

\Description{Construction of the gadget to summarize $\strempty$-paths.}
\caption{Construction of the gadget to summarize $\strempty$-paths.
  The SNFA in Figure~\ref{sfig:alg:qgraph-build:gadget:snfa} accepts strings of the form
  $(\query{q_1}^* \land \query{q_2})^*$.
  Assume for the sake of this example that $\epsilon$ is accepted by all queries.
  There is a path to $s_5$ after processing each character
  $w_i$ of the input string. The query context of $s_5$ is $[ \open(q_2), \open(q_1) ]$.
  Before processing the next character, $w_{i + 1}$ using the $s_5 \to s_5$ transition, the machine
  might make four choices:
  (\textit{a}) stay in the same state, or
  (\textit{b}) close and immediately reopen the query $q_1$, by following the sequence of
    transitions, $s_5 \to s_6 \to s_4 \to s_5$, or
  (\textit{c}) close and immediately reopen both queries: $s_5 \to s_6 \to s_7 \to s_2 \to^* s_5$,
    or
  (\textit{d}) move to the final state $s_8$.
  These alternatives are described in Figure~\ref{sfig:alg:qgraph-build:gadget:5e}. Observe how the
  gadget in Figure~\ref{sfig:alg:qgraph-build:gadget:qgraph} describes these and all other
  eventualities from the states of the SNFA.}
\label{sfig:alg:qgraph-build:gadget}
\end{figure}

%% file: src/figs/alg/epsilon-path.tex
\begin{figure}
\centering
\begin{tikzpicture}
  \node [circle, draw] (s0) at (0, 0) {};
  \node [circle, draw, label=above:{$[$}] (s1) at (1, 0) {};
  \node [circle, draw, label=above:{$]$}] (s2) at (2, 0) {};
  \node [circle, draw, fill=red!20, label=above:{$]$}] (s3) at (3, 0) {};
  \node [circle, draw] (s4) at (4, 0) {};
  \node [circle, draw, fill=red!20, label=above:{$]$}] (s5) at (5, 0) {};
  \node [circle, draw] (s6) at (6, 0) {};
  \node [circle, draw, fill=green!20, label=above:{$[$}] (s7) at (7, 0) {};
  \node [circle, draw] (s8) at (8, 0) {};
  \node [circle, draw, fill=green!20, label=above:{$[$}] (s9) at (9, 0) {};
  \node [circle, draw] (sa) at (10, 0) {};
  \node [circle, draw, label=above:{$[$}] (sb) at (11, 0) {};
  \node [circle, draw, label=above:{$]$}] (sc) at (12, 0) {};
  \node [circle, draw] (sd) at (13, 0) {};

  \draw [->] (s0) -- (s1) node {};
  \draw [->] (s1) -- (s2) node {};
  \draw [->] (s2) -- (s3) node {};
  \draw [->] (s3) -- (s4) node {};
  \draw [->] (s4) -- (s5) node {};
  \draw [->] (s5) -- (s6) node {};
  \draw [->] (s6) -- (s7) node {};
  \draw [->] (s7) -- (s8) node {};
  \draw [->] (s8) -- (s9) node {};
  \draw [->] (s9) -- (sa) node {};
  \draw [->] (sa) -- (sb) node {};
  \draw [->] (sb) -- (sc) node {};
  \draw [->] (sc) -- (sd) node {};

  \node [qgblank] (q0) at (0, -1) {};
  \node [qgblank] (q3) at (3, -1) {};
  \node [qgblank] (q51) at (4.5, -1) {};
  \node [qgblank] (q52) at (5.5, -1) {};
  \node [qgblank] (q7) at (7, -1) {};
  \node [qgblank] (q9) at (9, -1) {};
  \node [qgblank] (qd) at (13, -1) {};

  \draw [->] (q0) -- (q3) node [midway, above] {$E_{11}$};
  \draw [->] (q3) -- (q51) node [midway, above] {$E_{11}$};
  \draw [->] (q51) -- (q52) node [midway, above] {$E_{12}$};
  \draw [->] (q52) -- (q7) node [midway, above] {$E_{22}$};
  \draw [->] (q7) -- (q9) node [midway, above] {$E_{22}$};
  \draw [->] (q9) -- (qd) node [midway, above] {$E_{23}$};

  \draw [dashed] (s0) -- (q0) node {};
  \draw [dashed] (s3) -- (q3) node {};
  \draw [dashed] (s5) -- (q51) node {};
  \draw [dashed] (s5) -- (q52) node {};
  \draw [dashed] (s7) -- (q7) node {};
  \draw [dashed] (s9) -- (q9) node {};
  \draw [dashed] (sd) -- (qd) node {};

  \node (q0_bl) at ($(q0) + (0, -0.5)$) {};
  \node (q51_bl) at ($(q51) + (0, -0.5)$) {};
  \draw [decorate, decoration={brace, mirror}] (q0_bl) -- (q51_bl) node [midway, below] {Layer 1};
  \node (q52_bl) at ($(q52) + (0, -0.5)$) {};
  \node (q9_bl) at ($(q9) + (0, -0.5)$) {};
  \draw [decorate, decoration={brace, mirror}] (q52_bl) -- (q9_bl) node [midway, below] {Layer 2};
\end{tikzpicture}
\Description{Schematic procedure for how tentatively feasible $\strempty$-paths are broken down
  across the three layers of the gadget.}
\caption{Schematic procedure for how tentatively feasible $\strempty$-paths are broken down across
  the three layers of the gadget. All edges indicate $\epsilon$ transitions, and brackets, $]$ and
  $[$ indicate close and open nodes respectively. Unmatched have been shaded red and green. The
  edges in Layers~1 and~2 of the gadget summarize sequences of transitions between subsequent pairs
  of shaded nodes.}
\label{fig:alg:qgraph-build:gadget:eps}
\end{figure}

%% file: src/alg-qgraph-build-final.tex
\subsubsection{The Final Query Graph Construction}
\label{ssub:alg:qgraph-build:final}

Let $w = w_1 w_2 \cdots w_n$, and recall that $S$ is the set of states in the SNFA. We construct the
final query graph $G_M^w$ by tiling $n + 1$ copies of the gadget, and connecting each pair of
adjacent gadget instances, $i - 1$ and $i$, based on the transitions made by $M$ while processing
the character $w_i$:
\begin{equation}
\left.
\begin{aligned}
  G_M^w           & = (V, E, \qgstart, \qgend, \qgidx, l), \text{ where} \\
  V               & = S \times \{ 1, 2, 3 \} \times \{ 1, 2, \dots, n+1 \} \\
                  & = \{ (s, k, i) \mid (s, k) \in V_\Gad \text{ and } 1 \leq i \leq n+1 \}, \text{ and} \\
  E               & = E_c \union E_\strempty, \text{ where} \\
  E_c             & = \{ (s, 3, i) \to (s', 1, i+1) \mid
                         s, s' \in S \text{ and transition } s \to^{w_i} s' \text{ for } 1 \leq i \leq n \}, \\
  E_\strempty     & = \{ (s, k, i) \to (s', k', i) \mid
                         (s, k) \to (s', k') \in E_\Gad, 1 \leq i \leq n+1 \}, \text{ and} \\
  \qgstart        & = (s_0, 1, 1), \\
  \qgend          & = (s_f, 3, n+1), \\
  \qgidx(s, k, i) & = i, \text{ and} \\
  l(s, k, i)      & = l_\Gad(s, k).
\end{aligned}
\qquad \right\}
\label{eq:alg:graph-build:final}
\end{equation}


The following lemmas establish a connection between paths in $M$ and through $G_M^w$. Their proofs
may be found in Appendix~\refapp{app:alg:qgraph-build:final}.

\begin{makethm}{lemma}{lem:alg:qgraph-build:qg-pc1}[Path correspondence, Part 1]
Pick a path $\pi = s_0 \to^* s$ through $M$. If $\pi$ is tentatively feasible, then there is also a
tentatively feasible path $p = \qgstart \to^* (s, 3, |\str(\pi)|+1)$ through $G_M^w$ with the same
query context.
\end{makethm}

\begin{makethm}{lemma}{lem:alg:qgraph-build:qg-pc2}[Path correspondence, Part 2]
Pick a path $p = \qgstart \to^* v$ in $G_M^w$, where $v = (s, 3, i)$ for $s \in S$.
There is a tentatively well-parenthesized path $\pi = s_0 \to^* s$ in the SNFA with the same query
context, and such that if $p$ is tentatively feasible, then so is $\pi$.
\end{makethm}


From Lemma~\ref{lem:alg:qgraph-build:qg-pc2}, it follows that each path $p = \qgstart \to^* \qgend$
is well-parenthesized. It is also easy to see that string indices monotonically increase along each
edge $e \in E$, and that $G_M^w$ is a DAG. It follows that $G_M^w$ is a well-formed query graph.
The following theorem is now a direct consequence of Lemmas~\ref{lem:alg:qgraph-build:qg-pc1} and~%
\ref{lem:alg:qgraph-build:qg-pc2}, and establishes the correctness of the query graph:
\begin{theorem}
\label{thm:alg:qgraph-build:final}
$M$ accepts $w$ iff $\interp{G_M^w} = \true$.
\end{theorem}

%% file: src/alg-qgraph-eval.tex
\subsection{Evaluating Query Graphs Using Dynamic Programming}
\label{sub:alg:qgraph-eval}

Given a SemRE $r$ and a string $w$, we sequentially translate $r$ into an SNFA $M$ (using Figure~%
\ref{fig:alg:snfa}) and construct a query graph $G$ (using Equation~\ref{eq:alg:graph-build:final}).
Our final goal is now to evaluate $\interp{G}$.
Recall that all paths from $\qgstart$ to $\qgend$ are well-parenthesized, that a $\qgstart$--%
$\qgend$ path is feasible if the oracle accepts all of its queries, and that the entire graph $G$
evaluates to true if there is a feasible path. We now show how to determine this using dynamic
programming.

We start by introducing a few new properties of vertices in the query graph:
First, we say that a vertex $v$ in $G$ is \emph{alive} if there exists a tentatively feasible path
from $\qgstart$ to $v$.
Naturally:
\begin{theorem}
\label{thm:alg:qgraph-eval}
$\interp{G} = \true$ iff $\alive(\qgend) = \true$.
\end{theorem}

\input{src/figs/alg/qgraph-eval.tex}

We present inference rules for determining the living nodes of a query graph in Figure~%
\ref{fig:alg:qgraph-eval}. Rules~$\textsc{As}$,~$\textsc{Ab}$ and~$\textsc{Ao}$ (for the starting
node, blank and open nodes respectively) are immediate. On the other hand, upon encountering a close
node $v$ with $l(v) = \close(q)$, we need to discharge the queries that have just been fully
delimited. In principle: we must recall all tentatively feasible paths leading to $v$, find the last
unclosed open vertex on each of these paths, and discharge these queries. We do this efficiently by
also propagating the set of \emph{backreferences} during query graph evaluation.

Let $v$ be a vertex in $G$, and consider some tentatively feasible path $p = \qgstart \to^* v$. We
define the backreference of $p$, $\backref(p)$, as its last unclosed open vertex. The backreference
of the vertex $v$ is simply defined as the set of backreferences along each tentatively feasible
path to $v$:
\begin{alignat}{1}
  \backref(v) & = \{ \backref(p) \mid \text{Tentatively feasible path } p = \qgstart \to^* v \}.
\end{alignat}
Upon encountering a close node, $v$, we consult the backreferences $\backref(v')$ at each of its
immediate predecessors $v'$, and discharge the relevant queries. Formally, let $p = \qgstart \to^*
v$ be a tentatively feasible path with $l(v) = \close(q)$. We define $\aq(p)$ as the matching open
for this last vertex $v$. 
Aggregating this quantity along all tentatively feasible paths leading to $v$, we write:
\begin{alignat}{1}
  \aq(v) & = \{ \aq(p) \mid \text{Tentatively feasible path } p = \qgstart \to^* v \}.
\end{alignat}
To complete the calculation of $\alive(v)$, we simply note that if $v$ is a close node, then $v$ is
alive iff $\aq(v)$ is non-empty (see Rules~\textsc{M} and~\textsc{Ac}).

The remaining rules in Figure~\ref{fig:alg:qgraph-eval} focus on computing the backreferences. Once
again, the calculation for blank and open nodes is straightforward. See Rules~\textsc{Bb} and~%
\textsc{Bo}.
Lastly, let $v$ be a living close node, and consider how one might calculate its backreferences.
Informally, we have just popped the last open query from the query context, and we need to
``\emph{bring forward}'' the backreferences from each immediate predecessor of the popped open
vertices. This motivates the definition:
\begin{alignat}{1}
  \loq(v') & = \Union_{v'' \to v'} \backref(v''),
\label{eq:alg:qgraph-eval:loq}
\end{alignat}
and its use in Rule~\textsc{Bc}.
At this point, the following result is straightforward:
\begin{lemma}
The inference rules in Figure~\ref{fig:alg:qgraph-eval} are sound and complete.
\end{lemma}
One may therefore mechanically compute $\alive(v)$, $\backref(v)$ and $\aq(v)$ by repeatedly
applying the evaluation rules until fixpoint. This completes our description of the SemRE membership
testing procedure.

%% file: src/figs/alg/qgraph-eval.tex
\begin{figure}
\begin{center}
\mprset{sep=1.5em} 

\scalebox{0.75}{$
  \begin{array}{ccc}
    \inferrule*[Left=As]
               { }
               {\alive(\qgstart)}
    & \qquad \qquad \qquad &
    \inferrule*[Left=Ab]
               {l(v) = \blank \\
                v' \to v \\
                \alive(v')}
               {\alive(v)}
    \\[2em]
    \inferrule*[Left=Ao]
               {l(v) = \open(q) \\
                v' \to v \\
                \alive(v')}
               {\alive(v)}
    & &
    \inferrule*[Left=Ac]
               {l(v) = \close(q) \\
                \aq(v) \neq \emptyset}
               {\alive(v)}
  \end{array}
$}

\vspace{1.5em}

\scalebox{0.75}{$
    \inferrule*[Left=M]
               {l(v) = \close(q) \\
                v'\to v \\
                \alive(v') \\
                v'' \in \backref(v') \\
                \oracle(q, w_{\qgidx(v'')} \cdots w_{\qgidx(v)-1})}
               {v'' \in \aq(v)}
$}

\vspace{1.5em}

\scalebox{0.75}{$
  \begin{array}{ccc}
    \inferrule*[Left=Bb]
               {l(v) = \blank \\
                v' \to v \\
                \alive(v')}
               {\backref(v') \subseteq \backref(v)}
    & \qquad \qquad \qquad &
    \inferrule*[Left=Bo]
               {l(v) = \open(q) \\
                \alive(v)}
               {v \in \backref(v)}
  \end{array}
$}

\vspace{1em}

\scalebox{0.75}{$
    \inferrule*[Left=Bc]
               {l(v) = \close(q) \\
                \alive(v) \\
                v' \in \aq(v)}
               {\loq(v') \subseteq \backref(v)}
$}
\end{center}
\Description{Inference rules for evaluating the query graph.}
\caption{Inference rules for evaluating the query graph.
  $\loq(v')$ is defined in Equation~\ref{eq:alg:qgraph-eval:loq}.
  Rule~\textsc{Bc} has a flavor similar to the calculation of the grandparents of a vertex in a
  graph. Repeated application of this rule is responsible for the dominant $O(|r||w|^3)$ term in the
  final time complexity of our procedure. The rule is vacuous for the case of non-nested SemREs,
  because $\loq(v') = \emptyset$.}
\label{fig:alg:qgraph-eval}
\end{figure}

%% file: src/alg-analysis.tex
\subsection{Performance Analysis of Our Algorithm}
\label{sub:alg:analysis}

Let $r$ be a semantic regular expression, and $w \in \Sigma^*$ be the input string. Let $M_r$ be the
SNFA constructed using Figure~\ref{fig:alg:snfa}, and $G_{M_r}^w$ be the resulting query graph.
The proof of the following claim be found in Appendix~\refapp{app:alg:analysis}. In short, it
follows from bounding the sizes of the auxiliary sets, $\backref(v)$, $\aq(v)$ and $\loq(v)$ and the
observation that the query graph is a DAG and a subsequent accounting of the cost of applying each
of the inference rules from Figure~\ref{fig:alg:qgraph-eval} to each of the vertices in the query
graph.
\begin{makethm}{theorem}{thm:alg:analysis:main}
Assume the oracle responds in unit time.
We can determine whether $w \in \interp{r}$ in $O(|r|^2 |w|^2 + |r| |w|^3)$ time and with at most
$O(|r| |w|^2)$ calls to the oracle.
If $r$ does not contain nested queries, then the running time of our algorithm is $O(|r|^2 |w|^2)$.
Furthermore, if $\skel(r)$ is unambiguous, then the algorithm runs in $O(|r|^2 |w|)$ time and
discharges $O(|r| |w|)$ oracle queries.
\end{makethm}
Here, we write $\skel(r)$ (``\emph{skeleton}'') for the underlying classical regular expression,
which may be calculated by stripping away all oracle refinements occurring within the SemRE $r$.
Also recall that a (classical) regular expression $r$ is unambiguous if every string admits a single
parse tree~\cite{Book:IEEETransactions1971}.

%% file: src/complexity.tex
\section{The Complexity of Matching Semantic Regular Expressions}
\label{sec:complexity}

We will now establish some (admittedly simple) lower bounds on the complexity of the membership
testing problem. Our first result will show that any SemRE matching algorithm needs to make at least
quadratic (in the string length) number of oracle invocations in the worst case. Next, in Section~%
\ref{sub:complexity:triangles}, we will reduce the problem of finding triangles in graphs to that
of matching SemREs.


\subsection{Query Complexity}
\label{sub:complexity:query}

Note that we have so far not made any assumptions about the set of possible oracle queries $Q$. Our
first result can be stated in one of two forms, depending on whether $Q$ is finite or not:
\begin{makethm}{theorem}{thm:complexity:quad}
  If the set of queries $Q = \{ q_1, q_2, \dots \}$ is unbounded, then $\Omega(|r| |w|^2)$
  oracle invocations are necessary to determine whether $w \in \interp{r}$.
  If the query space $Q$ is finite but non-empty, then $\Omega(|w|^2)$ oracle invocations are
  necessary to determine whether $w \in \interp{r}$.
\end{makethm}
The proof in Appendix~\refapp{app:complexity} essentially fleshes out the observation that matching
$\Sigma^* \query{q} \Sigma^*$ requires us to examine \emph{every} substring of $w$.


\begin{note}
Note that Theorem~\ref{thm:complexity:quad} only describes the asymptotic growth of the worst-case
query complexity. On the other hand, specific SemREs might admit more efficient matching algorithms.
For example, whether a string $w$ matches $(\Sigma \land \query{q}) \Sigma^*$ can be determined with
only a single oracle query, by checking whether $\oracle(q, w_1) = \true$ (where $w_1$ is the first
letter of $w$).
\end{note}

\begin{note}
We also note that Theorem~\ref{thm:complexity:quad} only holds in the black-box setting when one is
not allowed to make any further assumptions about the oracle. In specific situations, such as when
the oracle is a set of strings or performs a computation that is easy to characterize, it might be
possible to perform additional optimizations so that these lower bounds no longer hold.
\end{note}

\mclearpage
\subsection{Membership Testing is as Hard as Finding Triangles}
\label{sub:complexity:triangles}

The gap between the worst case time complexity of the SNFA algorithm and the lower bound of
Theorem~\ref{thm:complexity:quad} is frustrating. We therefore conclude this section with a defense
of our algorithm.

Recall that the $O(|r||w|^3)$ term arises because of the time needed to compute the updated set of
backreferences at each close node in the query graph. This calculation is essentially multiplying
the Boolean vector listing backreferences from the current node with the matrix containing the
global set of backreferences. This observation led us to compare the SemRE membership testing
problem with problems in Boolean matrix multiplication and graph theory.

Let $G = (V, E)$ be an undirected graph with $V = \{ v_1, v_2, \dots, v_n \}$. We say that $G$ has a
\emph{triangle} if there are distinct vertices $v_i, v_j, v_k \in V$ with edges $\{ v_i, v_j \},
\{ v_j, v_k \}, \{ v_k, v_i \} \in E$. We will now construct a SemRE $r_\Delta$ and a string $w_G$
such that $w_G \in \interp{r_\Delta}$ exactly when $G$ has a triangle.

Let $\Sigma = \{ 1, 2, \dots, n, \# \}$ be an alphabet with $n + 1$ symbols, and where $\#$ is a
special delimiter symbol. Punning with the set of edges, let $E \in Q$ be a query such that for
non-empty strings $w = w_1 w_2 \dots w_k \in \Sigma^*$,
\[ \oracle(E, w) = \true \iff \{ w_1, w_k \} \in E. \]
Now consider the following SemRE,
\begin{alignat}{1}
  r_\Delta & = \Sigma^*
               \# \underbrace{(\Sigma \cdot
                               \underbrace{(\Sigma \Sigma^* \# \Sigma) \land \query{E}}_{r_{ij}} \cdot
                               \underbrace{(\Sigma \Sigma^* \# \Sigma) \land \query{E}}_{r_{jk}} \cdot
                               \Sigma) \land \query{E}}_{r_{ik}}
               \Sigma^*.
\label{eq:complexity:rdelta}
\end{alignat}
The marked subexpressions $r_{ij}$, $r_{jk}$ and $r_{ik}$ respectively identify portions of the
string that identify the endpoints of edges $\{ v_i, v_j \}$, $\{ v_j, v_k \}$, $\{ v_k, v_i \}$
respectively. From the semantics of $r_\Delta$ and the definition of $\oracle(E, w)$, it follows
that:
\begin{lemma}
\label{lem:complexity:triangles:interp}
Let $w_G = \#11 \#22 \#33 \cdots \#nn$. Then, $w_G \in \interp{r_\Delta}$ iff $G$ contains a
triangle.
\end{lemma}

One objection to this reduction might be the potentially large alphabet with $|\Sigma| = O(|V|)$.
This can be easily remedied by rendering the vertices of $V$ in binary, and rewriting $r_\Delta$ and
$w_G$ using the ternary alphabet $\Sigma_3 = \{ 0, 1, \# \}$. From here, it follows that:
\begin{theorem}
\label{thm:complexity:triangles}
Given a string $w$, if we can determine whether $w \in \interp{r_\Delta}$ in $D(|w|)$ time, then we
can determine whether a given graph $G$ contains triangles in $D(|V| \log(|V|))$ time.
\end{theorem}
There are triangle finding algorithms which run in $O(|V|^{3 - \epsilon})$ time. For example,
express the set of edges $E$ as an adjacency matrix and observe that $G$ contains a triangle iff the
trace (sum of diagonal elements) of $E^3$ is non-zero. Still, these procedures rely on algorithms
for fast (integer / Boolean) matrix multiplication, and are therefore impractical.
No purely ``\emph{combinatorial}'' algorithms are known that can find triangles in
$O(|V|^{3 - \epsilon})$ time~\cite{Williams2:FOCS2010, Yu2018:Triangles}. We conclude that finding
practical algorithms which perform SemRE membership testing in $O(|r|^c |w|^{3 - \epsilon})$ time
would be very challenging.

\begin{note}
Note the fundamental need for nested queries to express $r_\Delta$: $(\cdots \land \query{q_1}
\cdots) \land \query{q_2}$. The above reduction therefore does not contradict the promises of
Theorem~\ref{thm:alg:analysis:main} for non-nested SemREs. The Paris Hilton SemRE previously
mentioned in the introduction, $(\Sigma^* (\Sigma^* \land \query{\text{City}}) \Sigma^*) \land
\query{\text{Celebrity}}$, is one of the few natural examples of SemREs with nested queries that we
can write. We expect that most SemREs arising in practice would be unaffected by the reduction just
described.
\end{note}

%% file: src/eval.tex
\section{Experimental Evaluation}
\label{sec:eval}

We conducted an experimental evaluation of the algorithm from Section~\ref{sec:alg} in which we
attempted to answer the following questions:
\begin{enumerate}[label=\textbf{RQ\arabic*.}, ref=\textbf{RQ\arabic*}, leftmargin=*]
\item \label{enu:eval:throughput} How quickly can the algorithm process input data?
\item \label{enu:eval:oracle} How heavily does the algorithm use the oracle?
\item \label{enu:eval:input} How does the length of the input string affect matching speed?
\end{enumerate}

\paragraph{Benchmark SemREs and input data.}

We used the SemREs from Section~\ref{sub:semre:examples} to filter lines of text in a corpus of
spam emails and a corpus of Java code downloaded from three popular GitHub repositories. We removed
lines with non-ASCII characters and lines greater than 1,000 characters in length. After this
pruning the two datasets had 4,622,992 and 5,461,188 lines and measured 61~MB and 72~MB
respectively.
Table~\ref{tab:eval:benchmarks} presents some statistics of each SemRE, including their size, the
backing oracle, and the number of matched lines in the datasets. Because processing the entire
datasets with the LLM oracle would require several weeks of computation for $r_{\text{pass}}$,
$r_{\text{id}}$ and $r_{\text{spam, 1}}$, we only provide estimates of the number of matched lines
based on random sampling. Also note that we have padded the SemREs so the system looks for lines of
text with substrings that have the desired patterns.

\input{src/figs/eval/benchmarks.tex}

\paragraph{Implementation and baselines.}

We have implemented our algorithm in a system called $\grepo$, a grep-like tool for matching
semantic regular expressions.
We implemented a simple baseline algorithm based on dynamic programming, as discussed in Section~%
\ref{sub:semre:formalism}. This algorithm recursively processes subexpressions of the given SemRE
$r$ and substrings of the input string $w$, and uses top-down dynamic programming with memoization
to determine whether $w \in \interp{r}$.
The entire system consists of approximately 1,800 lines of Rust code, and has an interface similar
to grep: given the SemRE, oracle, and input file, it prints the matched lines.

Implementing the underlying oracles required some care: For example, the external Whois service
will temporarily block the users if too many queries are submitted in a short period of time. To
avoid overloading this service, we instead prepopulated a local database with necessary information
for all domains in the input data. We also provided a predefined list of phishing websites, an IP
geolocation database, and a filesystem interface. For SemREs using an LLM oracle, we used LLaMa3-8B
with 8-bit quantization~\cite{LlaMa3, LlaMa3-8bit}. Because the throughput of the LLM is low, and to
mitigate issues with nondeterminism (Recall Assumption~\ref{assm:semre:llm}), its use was mediated
by a query cache.

\mclearpage
\input{src/eval-throughput.tex} \mclearpage
\input{src/eval-oracle.tex}     \mclearpage
\input{src/eval-input.tex}      \mclearpage

%% file: src/figs/eval/benchmarks.tex
\begin{table}
\caption{Benchmark SemREs and their statistics. For the $r_{\text{id}}$ SemRE, the regular
  expressions $\text{pad}_1$ and $\text{pad}_2$ are added as padding to search for identifiers in
  longer lines of code. Here: $\text{pad}_1 = (\Sigma^* (\Sigma \backslash \Sigma_l))?$ and
  $\text{pad}_2 = (\Sigma^* (\Sigma \backslash (\Sigma_l \cup \{0,\dots,9\})))?$.}
\label{tab:eval:benchmarks}
\centering
\begin{tabular}{llllll}
\toprule
Dataset  & Name   & SemRE                 & Oracles               & $|r|$ & \# Matched line \\
\midrule
Code     & pass   & $\padded{pass}$       & LLM                   & 26    & $\approx$ 4,330 \\
         & file   & $\padded{file}$       & File system           & 226   & 8,706 \\
         & id     & $\text{pad}_1 \, r_\text{id} \, \text{pad}_2$ & LLM                   & 250   & $\approx$ 423,165 \\
\midrule
Spam     & edom   & $\padded{edom}$       & Whois                 & 307   & 365,112 \\
         & spam,1 & $\padded{spam,1}$     & LLM                   & 31    & $\approx$ 46,666 \\
         & spam,2 & $\padded{spam,2}$     & LLM                   & 86    & 21,096 \\
         & wdom,1 & $\padded{wdom,1}$     & Phishing website list & 263   & 192 \\
         & wdom,2 & $\padded{wdom,2}$     & Whois                 & 263   & 482  \\
         & ip     & $\padded{ip}$         & IP geolocation        & 153   & 135,320 \\
\bottomrule
\end{tabular}
\end{table}

%% file: src/eval-throughput.tex
\subsection{\ref{enu:eval:throughput}: Efficiency in Processing Input Data}
\label{sub:eval:throughput}

To measure the effectiveness of our algorithm in processing semantic regular expressions, we
compared its throughput and that of the baseline DP-based implementation for each SemRE from Table~%
\ref{tab:eval:benchmarks}. We present the results in Table~\ref{tab:eval:throughput}. For each
semantic regular expressions, we report the reciprocal throughput (average time needed to process
each line) for both the entire run and when restricted to lines that match $r$. Because of the low
throughput of the baseline and of the LLM-based oracle, we set a timeout of 40 minutes for each run
and report statistics obtained for the text processed within this period.

\input{src/figs/eval/throughput.tex}

Overall, we observe that $\grepo$ has a much higher throughput both for the file as a whole, and
also when restricted to matched lines. Although this speedup varies depending on the SemRE in
question, averaged across all SemREs (using geometric means), and compared to the DP-based baseline,
$\grepo$ has $101\times$ the throughput for the whole dataset, and $12\times$ the throughput when
focusing on the matched lines in particular.

We notice that this speedup is lower when the oracle that has a very long response time (i.e., for
the LLM-based oracle). This is unsurprising, and is caused by the fact that performance for these
SemREs is dominated by the response time of the oracle, as we will see in the next section.
Nevertheless, our algorithm still outperforms the baseline in all but one of the LLM-based
SemREs.
The only exception is $r_\text{spam, 1}$ where the DP-baseline has a slightly higher
throughput. We speculate that this performance inversion is due the underlying skeleton being too
permissive: matching essentially reduces to guessing which substring is accepted by the LLM, so that
the simplicity of the baseline algorithm results in reduced overhead.

%% file: src/figs/eval/throughput.tex
\begin{table}
\caption{SemRE matching performance. We report reciprocal throughput (average time needed to process
  each line) for the two algorithms, and aggregate oracle use statistics.}
\label{tab:eval:throughput}
\centering
\resizebox{\textwidth}{!}{
\begin{tabular}{ccccccccccc}
\toprule
  & \multicolumn{2}{c}{RT, Total ($\text{ms} \cdot \text{line}^{-1}$)}
  & \multicolumn{2}{c}{RT, Matched ($\text{ms} \cdot \text{line}^{-1}$)}
  & \multicolumn{2}{c}{Oracle calls ($\text{line}^{-1}$)}
  & \multicolumn{2}{c}{Oracle fraction}
  & \multicolumn{2}{c}{Query length ($\text{chars} \cdot \text{line}^{-1}$)}
\\
  \cmidrule(lr){2-3}
  \cmidrule(lr){4-5}
  \cmidrule(lr){6-7}
  \cmidrule(lr){8-9}
  \cmidrule(lr){10-11}
SemRE  &  SNFA & DP &  SNFA & DP &  SNFA & DP &  SNFA & DP &  SNFA & DP\\\midrule
pass          & \textbf{38.2}     & 464.2            & \textbf{891.0}    & 3624.0             & \textbf{0.125}  & 1.767             &  0.999 &  0.999         & \textbf{1.93}    & 23.895 \\
file          & \boldsmall        & 1.7              & \textbf{0.1}      & 5.5                & 0.001           & 0.001             &  0.001 &  $\le0.001$    & \textbf{0.033}   & 0.036 \\
id            & \textbf{175.7}    & 1569.7           & \textbf{311.0}    & 1545.5             & \textbf{3.101}  & 20.541            &  0.999 &  0.999         & \textbf{22.218}  & 111.249 \\
ip            & \boldsmall        & 159.3            & \textbf{0.1}      & 10.8               & 0.113           & \textbf{0.095}    &  0.162 &  $\le0.001$     & 1.205            & \textbf{1.016} \\
edom          & \boldsmall        & 16.9             & \textbf{0.3}      & 2.9                & 0.09            & \textbf{0.082}    &  0.417 &  $\le0.001$     & 0.993            & \textbf{0.651} \\
spam,1        & 7.9               & \textbf{6.3}     & 763.5             & \textbf{525.0}     & \textbf{0.032}  & 0.037             &  0.999 &  0.891          & 0.892            & \textbf{0.386} \\
spam,2        & \textbf{0.4}      & 23.5             & \textbf{39.5}     & 1728.1             & \textbf{0.016}  & 0.117             &  0.987 &  0.960          & \textbf{0.076}   & 0.415 \\
wdom,1        & \textbf{0.2}      & 35.9             & \textbf{0.4}      &                    & 0.424           & \textbf{0.186}    &  0.925 &  0.003          & 6.15             & \textbf{2.616} \\
wdom,2        & \boldsmall        & 2.6              & \textbf{0.4}      & 7.3                & \textbf{0.025}  & 0.055             &  0.351 &  0.002          & \textbf{0.24}    & 0.523 \\
\bottomrule
\end{tabular}}
\end{table}

%% file: src/eval-oracle.tex
\subsection{\ref{enu:eval:oracle}: Extent of Oracle Use}
\label{sub:eval:oracle}

We expect oracle use to have non-zero cost (either in terms of running time, actual cost, or in
terms of limits on daily use). We therefore measured how heavily the two algorithms used the oracle
in terms of number of calls, time spent within the oracle, and average length of submitted queries.
Once again, we report these statistics in Table~\ref{tab:eval:throughput}.
%
Overall, $\grepo$ makes $51\%$ fewer oracle calls than the baseline, which in turn spends $3 \times$
the time inside the oracle as our algorithm.

These additional ``useless'' queries are due to a combination of factors:
One reason is because of lines where the skeleton eventually turns out to not match. Recall that the
skeleton is the classical regular expression $\skel(r)$ that one may obtain by removing oracle
refinements. Naturally, $\interp{r} \subseteq \interp{\skel(r)}$, so any oracle use for strings
$w \notin \interp{\skel(r)}$ is wasted work.
A second reason is due to a lack of lookahead even for lines where the skeleton matches. For
example, in the case of the spam medication SemRE, while the baseline would correctly use
whitespaces on the left to demarcate the beginning of medicine names, it would not be able to
exploit the symmetric property of whitespace demarcation on the right.

\begin{note}
As a counterpoint to our entire paper, a critical reader might wonder whether it would make sense to
simply submit the entire string to the LLM and suitably phrase the SemRE matching problem as a
prompt. Of course, this approach would only work when the backing oracle is an LLM, and would lack
the semantic guarantees of SemREs. We additionally point out that even for LLM-backed SemREs,
$\grepo$ submits much fewer tokens to the oracle than the average line length, so LLM use is
massively reduced.
\end{note}

%% file: src/eval-input.tex
\subsection{\ref{enu:eval:input}: Scalability with Respect to Input String Length}
\label{sub:eval:input}

Finally, we investigate the empirical impact of string length on the throughput of our algorithm. We
reran the experiment on the same dataset but filtered out all lines longer than 200 characters. We
show the median throughput of the matching algorithms as a function of line length in Figure~%
\ref{fig:eval:input}. We also show the overall distribution of line lengths in our dataset.

\begin{figure}
\centering
\includegraphics[width=.9\textwidth]{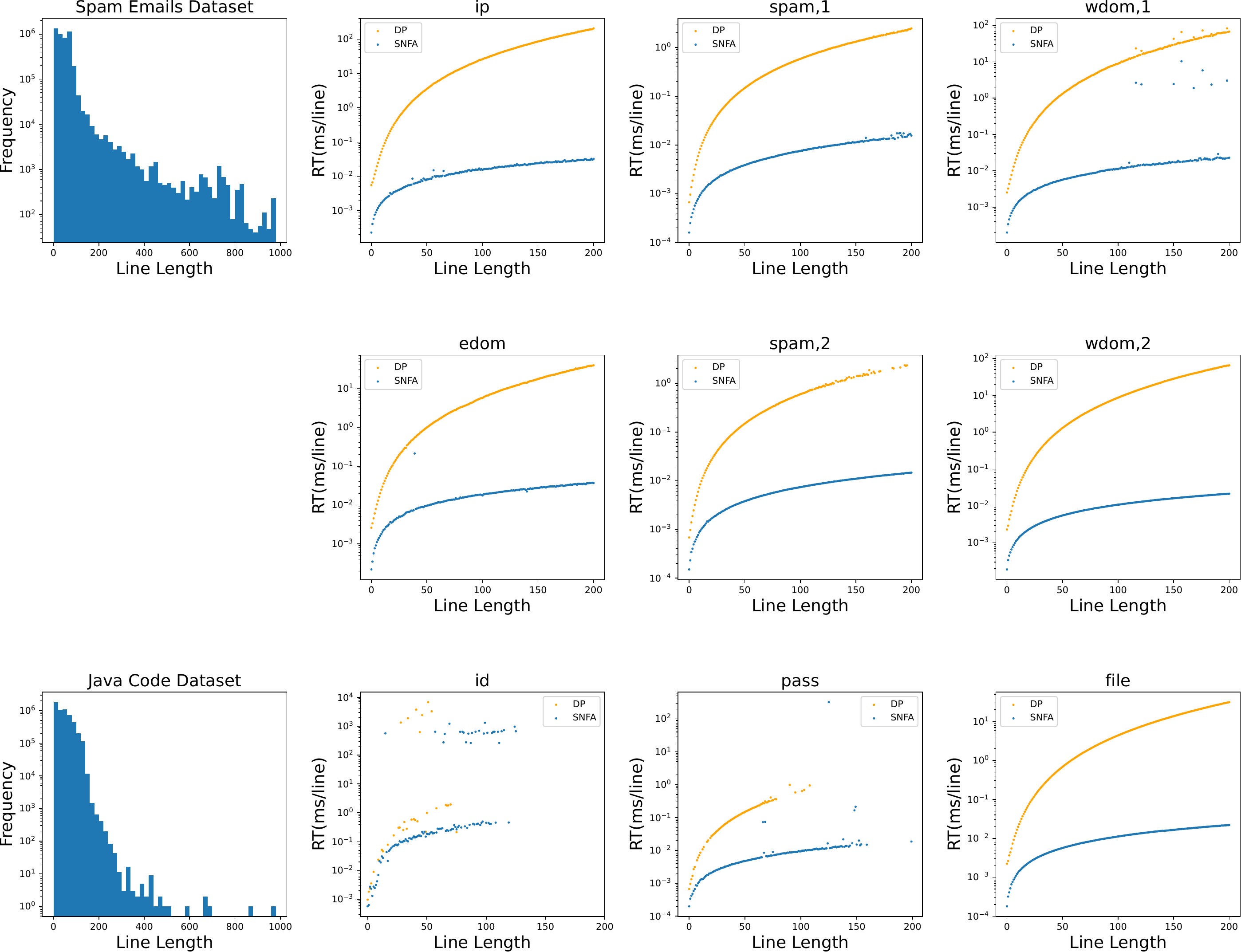}
\Description{Line length distribution and variation of running time across the datasets.}
\caption{Line length distribution and variation of running time across the datasets. We only report
  running time measurements when the algorithm processes at least 10 lines with that length.}
\label{fig:eval:input}
\end{figure}

Given the overall differences in throughput, it is unsurprising that the median runtime of $\grepo$
increases much more slowly compared to the baseline (notice the log-scaling applied to the
$y$-axis). We observe similar trends both for LLM and non LLM oracles.
The ``\emph{cloud}'' of data points appearing to the top right of the plots for $r_{\text{id}}$
occur because of cache misses in the oracle implementation leading to actual LLM queries. This cloud
of cache misses is particularly acute for this SemRE because identifiers appear in a large fraction
of lines.

%% file: src/related.tex
\section{Related Work}
\label{sec:related}


\paragraph{SMORE}

The story of our paper begins with SMORE~\cite{Smore}, which is itself set in a broader context of
systems that use LLMs for various data extraction~\cite{FlashGPT, Cheng:Binder:ICLR2023} and code
generation tasks~\cite{Codegen}. Other examples of work in this space include $L^*LM$~%
\cite{LstarLM}, which uses an LLM as oracle to guide the $L^*$ algorithm (Angluin's famous procedure
for learning DFAs~\cite{Lstar}) and systems that extract structured data from language models~%
\cite{Prompting, relm}.

The overarching difference between SMORE and our paper is in the focus: \citet{Smore} consider the
problem of synthesizing semantic regular expressions from data, rather than membership testing.
Their implementation uses the dynamic programming approach for membership testing that we used as a
baseline in Section~\ref{sec:eval}. This is a natural approach to testing membership: an early
reference is~\cite{HopcroftU1979}.
In addition, SemREs, as defined by Equations~\ref{eq:semre:syntax} and~\ref{eq:semre:semantics},
appear to be simpler than the SMORE formalism (Figures~3 and~4), in a few ways:
\begin{enumerate}
\item SMORE-style SemREs appear to permit string transformations during processing by the oracle.
  This allows languages such as $\{ w \in \Sigma^* \mid w \text{ is a date and }
  \operatorname{month}(w) = \text{May} \}$ to be represented. Because the backing oracle $\oracle$
  is externally defined and we make no assumptions about the query space $Q$, they can both be
  suitably extended so that such languages remain expressible in our simpler core formalism.
\item In addition to the basic constructors, SMORE-style SemREs also permit negation, intersection,
  and bounded repetition. Efficient membership testing for languages defined using these operations
  is significantly harder than for classical regular expressions. See, for
  example,~\cite{KupfermanZuhovitzky:MFCS2002, Rosu2007} and~\cite{BVAScan}. To keep our paper
  simple, we have ignored these constructs. Handling them would be an excellent direction for future
  work.
\end{enumerate}
On the other hand, our work is also more general, because the oracle can be realized using any
external mechanism. While oracle machines have been the subject of extensive theoretical
study, to the best of our knowledge, our paper is the first practical implementation of such
ideas.


\paragraph{KAT}

Another closely related idea is that of KAT (Kleene algebras with tests)~\cite{Kozen1997KAT}. KAT
has a similar flavor as SemREs, with the major difference being in the nature of oracle tests: While
KAT allows individual characters of a string to encode uninterpreted tests, SemREs allow entire
substrings to be presented to an oracle for validation.


\paragraph{Visibly pushdown languages.}

The well-nested query structure induced by SemREs is reminiscent of visibly pushdown languages~%
\cite{Rajeev-VPL:STOC2004}. We highlight some differences:
First, visibly pushdown automata (VPAs) require an explicitly tagged alphabet of calls, returns, and
local actions. As a consequence, the locations where the stack is pushed an popped is obvious
(i.e., ``\emph{visible}'') upon seeing the input string. On the other hand, identifying positions
where queries are opened and closed is the key computational problem associated with SemREs. Think
of the $\Sigma^* \query{\text{Politician}} \Sigma^*$ pattern. As a consequence, for a fixed VPA $M$,
we only require $O(|w|)$ time to determine whether $M$ accepts an input string $w$. This is in
opposition to our $\Omega(|w|^2)$ lower bound of Theorem~\ref{thm:complexity:quad}.
A second difference is that VPAs allow for unbounded call depths, whereas the query context in an
SNFA can never be deeper than the nesting depth.


\paragraph{Membership testing for classical and extended regexes}

The problem of membership testing for regular expressions has also been extensively studied. We
refer the reader to \citeauthor{Cox2010Regex}'s excellent survey on the topic~\cite{Cox2010Regex}.
Practical implementations use a range of techniques, including backtracking~\cite{PCRE}, automata
simulation~\cite{RE2, WangHCPLHZ2019Hyperscan, Thompson1968}, and various forms of derivatives~%
\cite{Brzozowski1964, Antimirov1996}.
Extended regular expressions, i.e., with intersection and complement, are of particular interest to
us, because of the similarities between intersection, $r_1 \intersection r_2$, and the way oracle
access is mediated in SemREs: $r \land \query{q}$. It is now folklore that these operations do not
add expressive power to classical regular expressions. Still, efficient membership testing
algorithms for this class is challenging, as extended regular expressions are non-elementarily more
succinct than NFAs: We refer the reader to~\cite{Rosu2007}.

\citet{Petersen:STACS2002} showed that the problem of regular expressions with intersection is
LOGCFL-complete. The reduction in Section~\ref{sub:complexity:triangles} shows SemRE membership
testing is similarly related to several other problems of interest, including triangle finding,
Boolean matrix multiplication~\cite{Williams2:FOCS2010}, context-free grammar (CFG) parsing~%
\cite{lee2002fast}, and Dyck reachability~\cite{Dyck-survey, Cetti-Hansen}. These connections also
raise the possibility of applying techniques such as the Four Russians Method~\cite{four-russians}
or ideas similar to Valiant's reduction of CFG parsing to Boolean matrix multiplication~%
\cite{valiant-cfg} to develop asymptotically faster algorithms. Of course, this reduction does not
work well in practice, so developing practical subcubic algorithms would be challenging.


\paragraph{Derivative-based algorithms.}

Derivatives~\cite{Brzozowski1964, Antimirov1996} provide an alternative characterization of regular
expression matching algorithms. Given a regular expression $r$ and an initial character $a$, the
(Brzozowski) derivative $\partial_a(r)$ is another regular expression that describes possible
suffixes that extend this initial character: $\interp{\partial_a(r)} = \{ w' \mid aw' \in
\interp{r} \}$. For example, it can be seen that $\partial_a((ab)^*) = b(ab)^*$.
There is a purely syntactic approach to calculating derivatives, and this greatly simplifies
implementation. Naturally, one can ask whether similar derivative-based approaches can be applied to
the setting of SemREs. We think that this is a good direction of future work. They might even allow
us to gracefully handle extended regex operators, in a manner similar to~%
\cite{Veanes-Derivatives:POPL2025}.
However, while derivatives have promise to simplify some parts of our algorithm, we anticipate the
core difficulties---i.e., nested queries and associated hardness results---will remain.



%% file: src/concl.tex
\section{Conclusion}
\label{sec:concl}

In this paper, we studied the membership testing problem for semantic regular expressions (SemREs).
We described an algorithm, analyzed its performance both theoretically and experimentally, and
proved some lower bounds and hardness results. We think that SemREs are an exciting formalism, and
that their applications and associated computational problems deserve greater scrutiny.

%% file: src/artifact.tex
\section*{Artifact Availability Statement}

The artifact that supports the findings of this paper can be downloaded from Zenodo~\cite{Artifact}.

%% file: src/acks.tex
\section*{Acknowledgements}

This research was supported in part by the US National Science Foundation awards CCF~2313062, CCF~%
2146518 and CCF~2107261. We thank our anonymous reviewers for their feedback, which helped to
greatly improve this paper.

%% file: app/alg.tex
\section{Matching Semantic Regular Expressions Using Query Graphs}
\label{app:alg}

\input{app/alg-snfa.tex}          \mclearpage
\input{app/alg-qgraph-build.tex}  \mclearpage
\input{app/alg-analysis.tex}

%% file: app/alg-snfa.tex
\subsection{Converting SemREs into Semantic NFAs}
\label{app:alg:snfa}


We formalize Figure~\refshort{fig:alg:snfa} with the following recursive construction of $M_r$ given
$r$:
\begin{enumerate}
\item When $r = \bot$: $M_\bot = (\{ s_0, s_f \}, \emptyset, L_\bot, s_0, s_f)$, with $L_\bot(s_0)
  = L_\bot(s_f) = \blank$.
\item When $r = \strempty$: $M_\strempty = (\{ s_0, s_f \}, \{ s_0 \to^\strempty s_f \},
  L_\strempty, s_0, s_f)$, with $L_\strempty(s_0) = L_\strempty(s_f) = \blank$.
\item When $r = a$ for $a \in \Sigma$: $M_a = (\{ s_0, s_f \}, \{ s_0 \to^a s_f \}, L_a, s_0, s_f)$,
  with $L_a(s_0) = L_a(s_f) = \blank$.
\item When $r = r_1 + r_2$: $M_r = (S_r, \Delta_r, L_r, s_0, s_f)$, where:
  \begin{alignat*}{1}
  S_r      & =      \{ s_0, s_f \} \union
                    \{ (1, s) \mid s \in S_{r_1} \} \union
                    \{ (2, s) \mid s \in S_{r_2} \}, \\
  \Delta_r & =      \{ s_0 \to^\strempty (1, s_{10}), s_0 \to^\strempty (2, s_{20}) \} \union {} \\
           & \qquad \{ (1, s) \to^a (1, s') \mid s \to^a s' \in \Delta_{r_1} \} \union {} \\
           & \qquad \{ (2, s) \to^a (2, s') \mid s \to^a s' \in \Delta_{r_2} \} \union {} \\
           & \qquad \{ (1, s_{1f}) \to^\strempty s_f, (2, s_{2f}) \to^\strempty s_f \},
                    \text{ and} \\
  L(s)     & = \begin{cases}
               \blank       & \text{if } s = s_0 \text{ or } s = s_f, \\
               L_{r_1}(s_1) & \text{if } s = (1, s_1) \text{ for } s_1 \in S_{r_1}, \text{ and} \\
               L_{r_2}(s_2) & \text{if } s = (2, s_2) \text{ for } s_2 \in S_{r_2},
               \end{cases}
  \end{alignat*}
  where $M_{r_1} = (S_{r_1}, \Delta_{r_1}, L_{r_1}, s_{10}, s_{1f})$ and $M_{r_2} = (S_{r_2},
  \Delta_{r_2}, L_{r_2}, s_{20}, s_{2f})$ are the SNFAs for each subexpression.
\item When $r = r_1 r_2$: $M_r = (S_r, \Delta_r, L_r, s_0, s_f)$, where:
  \begin{alignat*}{1}
  S_r      & =      \{ s_0, s_f \} \union
                    \{ (1, s) \mid s \in S_{r_1} \} \union
                    \{ (2, s) \mid s \in S_{r_2} \}, \\
  \Delta_r & =      \{ s_0         \to^\strempty (1, s_{10}),
                       (1, s_{1f}) \to^\strempty (2, s_{20}),
                       (2, s_{2f}) \to^\strempty s_f \} \union {} \\
           & \qquad \{ (1, s) \to^a (1, s') \mid s \to^a s' \in \Delta_{r_1} \} \union {} \\
           & \qquad \{ (2, s) \to^a (2, s') \mid s \to^a s' \in \Delta_{r_2} \} \text{ and} \\
  L(s)     & = \begin{cases}
               \blank     & \text{if } s = s_0 \text{ or } s = s_f, \\
               L_{r_1}(s_1) & \text{if } s = (1, s_1) \text{ for } s_1 \in S_{r_1}, \text{ and} \\
               L_{r_2}(s_2) & \text{if } s = (2, s_2) \text{ for } s_2 \in S_{r_2}.
               \end{cases}
  \end{alignat*}
  As before, $M_{r_1} = (S_{r_1}, \Delta_{r_1}, L_{r_1}, s_{10}, s_{1f})$ and $M_{r_2} = (S_{r_2},
  \Delta_{r_2}, L_{r_2}, s_{20}, s_{2f})$ are the subexpression SNFAs.
\item When $r = r_1^*$: $M_r = (S_r, \Delta_r, L_r, s_0, s_f)$, where:
  \begin{alignat*}{1}
  S_r      & =      \{ s_0, s_f \} \union \{ (1, s) \mid s \in S_{r_1} \}, \\
  \Delta_r & =      \{ s_0 \to^\strempty (1, s_{10}),
                       (1, s_{1f}) \to^\strempty s_0,
                       s_0 \to^\strempty s_f \} \union {} \\
           & \qquad \{ (1, s) \to^a (1, s') \mid s \to^a s' \in \Delta_{r_1} \}, \\
  L(s)     & = \begin{cases}
               \blank     & \text{if } s = s_0 \text{ or } s = s_f, \\
               L_{r_1}(s_1) & \text{if } s = (1, s_1) \text{ for } s_1 \in S_{r_1}.
               \end{cases}
  \end{alignat*}
  Once again, $M_{r_1} = (S_{r_1}, \Delta_{r_1}, L_{r_1}, s_{10}, s_{1f})$ is the recursively
  constructed SNFA.
\item For $r = r_1 \land \query{q}$: $M_r = (S_r, \Delta_r, L_r, s_0, s_f)$, where:
  \begin{alignat*}{1}
  S_r      & =      \{ s_0, s_f \} \union \{ (1, s) \mid s \in S_{r_1} \}, \\
  \Delta_r & =      \{ s_0 \to^\strempty (1, s_{10}), (1, s_{1f}) \to^\strempty s_f \} \union {} \\
           & \qquad \{ (1, s) \to^a (1, s') \mid s \to^a s' \in \Delta_{r_1} \}, \\
  L(s)     & = \begin{cases}
               \open(q)     & \text{if } s = s_0, \\
               \close(q)    & \text{if } s = s_f, \text{ and} \\
               L_{r_1}(s_1) & \text{if } s = (1, s_1) \text{ for } s_1 \in S_{r_1}.
               \end{cases}
  \end{alignat*}
  $M_{r_1}$ is once again the recursively constructed SNFA.
\end{enumerate}
We can show by induction on $r$ that all paths from $s_0$ to $s_f$ have well-parenthesized state
labels. It follows that $M_r$ is a well-formed SNFA.

%% file: app/alg-qgraph-build.tex
\subsection{Constructing the Query Graph}
\label{app:alg:qgraph-build}

Some assumptions that will simplify our reasoning in the rest of this section:
\begin{assumption}
\label{assm:alg:qgraph-build:snfa2}
\begin{enumerate}
\item The start state $s_0$ of $M$ is $\blank$. If not, without loss of generality, we can create a
  new start state $s_0'$ with an $\strempty$-transition, $s_0' \to^\strempty s_0$.
\item The target $s'$ of any character transition, $s \to^a s'$ for $a \in \Sigma$, has a blank
  label, i.e., $\lambda(s') = \blank$. If not, we once again introduce a new state $s''$ with
  $\lambda(s'') = \blank$, and replace the previous transition with: $s \to^a s'' \to^\strempty s'$.
\end{enumerate}
The last two assumptions are needed to resolve technical complications that arise while constructing
the gadget, and will only be referenced in the appendix in the proof of Lemma~%
\refshort{lem:alg:qgraph-build:qg-pc1}.
We elide these assumptions in the rest of this paper.
\end{assumption}

\input{app/alg-qgraph-build-prelims.tex}    \mclearpage
\input{app/alg-qgraph-build-gadget.tex}     \mclearpage
\input{app/alg-qgraph-build-final.tex}

%% file: app/alg-qgraph-build-prelims.tex
\subsubsection{Preliminaries: Tentatively Feasible Subpaths, Query Contexts, and
  $\strempty$-Feasibility.}
\label{app:alg:qgraph-build:prelims}

We present inference rules to identify $\strempty$-feasible paths through the SNFA in Figure~%
\ref{fig:alg:qgraph-build:eps}. Notice that they are simply an alternative version of the rules for
computing transitive closures:
\begin{lemma}
\label{lem:alg:qgraph-build:eps}
The inference rules in Figure~\ref{fig:alg:qgraph-build:eps} are sound and complete for determining
whether $s \efp s'$. Moreover, they can be evaluated to fixpoint in $O(|r|^2)$ time.
\end{lemma}
Because paths through the SNFA are (tentatively) well-parenthesized, for each pair of derivable
judgments, $s \to s' \mid \sigma_o$ and $s \to s' \mid \sigma_o'$, it must be the case that
$\sigma_o = \sigma_o'$. In other words, the set of outstanding queries that need to be opened is
uniquely determined by the initial and final states of the path $s$ and $s'$. Therefore, there are
at most $O(|r|^2)$ judgments of the form $s \to s' \mid \sigma_o$. We evaluate them to fixpoint
using a simple DFS-based algorithm.

\begin{figure}
\begin{center}
\mprset{sep=1.5em} 

\scalebox{0.75}{$
  \begin{array}{ccccc}
    \inferrule*{s \efp s' \mid []}
               {s \efp s'}
    & \qquad \qquad \qquad &
    \inferrule*{\lambda(s) = \blank}
               {s \efp s \mid []}
    & \qquad \qquad \qquad &
    \inferrule*{\lambda(s) = \open(q)}
               {s \efp s \mid [q]}
  \end{array}
$}

\vspace{1.5em}

\scalebox{0.75}{$
  \begin{array}{ccc}
    \inferrule*{s \efp s' \mid \sigma_o \\
                s' \to^\strempty s'' \\
                \lambda(s'') = \blank}
               {s \efp s'' \mid \sigma_o}
    & \qquad \qquad \qquad &
    \inferrule*{s \efp s' \mid \sigma_o \\
                s' \to^\strempty s'' \\
                \lambda(s'') = \open(q)}
               {s \efp s'' \mid q :: \sigma_o}
  \end{array}
$}

\vspace{1.5em}

\scalebox{0.75}{$
    \inferrule*{s \efp s' \mid q :: \sigma_o \\
                s' \to^\strempty s'' \\
                \lambda(s'') = \close(q) \\
                \oracle(q, \strempty) = \true}
               {s \efp s'' \mid \sigma_o}
$}
\end{center}
\Description{Inference rules to identify $\strempty$-feasible paths in the SNFA, $s \efp s'$.}
\caption{Inference rules to identify $\strempty$-feasible paths in the SNFA, $s \efp s'$. The
  auxiliary judgment $s \efp s' \mid \sigma_o$ indicates the presence of an tentatively feasible
  $\strempty$-path from $\pi = s \to^* s'$ with no unmatched close nodes, and with $\sigma_o$ being
  the stack of outstanding open queries.}
\label{fig:alg:qgraph-build:eps}
\end{figure}

%% file: app/alg-qgraph-build-gadget.tex
\subsubsection{The Inter-Character Gadget}
\label{app:alg:qgraph-build:gadget}

The following result establishes a connection between paths through the gadget and through the SNFA.
It is needed for the proof of Lemma~\refshort{lem:alg:qgraph-build:qg-pc2}, which is a similar path
correspondence, but for the entire query graph.

\begin{lemma}[Path correspondence for gadgets]
\label{lem:qgraph-build:gadget:pc}
Pick a path $p = (s, 1) \to^* (s', k')$ through $\Gad$, for $s, s' \in S, \lambda(s) = \blank$ and
$k' \in \{ 1, 2, 3 \}$. Then, there is a tentatively feasible path through the SNFA $M$ from $s$ to
$s'$ with the same query context when any of the following conditions hold:
\begin{enumerate}
\item if $\lambda(s') = \blank$ or $\lambda(s') = \close(q)$ for $q \in Q$, or if
\item if $\lambda(s') = \open(q)$ for $q \in Q$ and $k' \geq 2$.
\end{enumerate}
\end{lemma}
\begin{proof}
We prove this by induction over the path $p$. The base case is immediate. For the induction step,
assume that the result is known for a path $p$, and that our goal is to prove the result for a
one-step extension, $p'$:
\[ p' = \underbrace{(s, 1) \to^* (s', k')}_{p} \to (s'', k''). \]
We analyze the value of $k''$:
\begin{enumerate}
\item \label{enu:qgraph-build:gadget:pc:kpp3}
  When $k'' = 3$: It has to be the case that $k' = 2$. We use the induction hypothesis to recover a
  tentatively feasible path $\pi = s \to^* s'$. Notice that the newly added edge $e = (s', k') \to
  (s'', k'')$ is either witnessed by an $\strempty$-labelled feasible path $\pi_e = s''' \to^* s''$
  through the SNFA with $s' \to^{\epsilon} s'''$ or satisfies $s' = s''$. In the first case we
  construct the combined path
  \[ \pi' = s \to^* s' \to^{\epsilon} s''' \to^* s'', \]
  by concatenating $\pi$ and $\pi_e$. In the later case we simply take $\pi' = \pi$.
  In any case $\pi'$ remains tentatively feasible, with the same query context.
\item \label{enu:qgraph-build:gadget:pc:kpp2}
When $k'' = 2$: We perform a nested case analysis on $k'$:
  \begin{enumerate}
  \item \label{enu:qgraph-build:gadget:pc:kpp3:kp1}
    When $k' = 1$: Two conclusions follow immediately: $\lambda(s') \neq \open(q)$ for any $q \in
    Q$, and $s'' = s'$. Our witness to the extended path $p'$ is now simply the witness $\pi$ to the
    previous path $p$.
  \item \label{enu:qgraph-build:gadget:pc:kpp3:kp2}
    When $k' = 2$: It follows immediately that $\lambda(s'') = \open(q'')$. Apply the induction 
    hypothesis and we recover a tentatively feasible path $\pi = s \to^* s'$. From the construction 
    of the gadget, either there is a path:
    \[ \pi'' = s' \to^\strempty \underbrace{s'''' \to^* s'''}_{\pi'''} \to^\strempty s'', \]
    where the $\pi'''$ is an (unconditionally) feasible
    $\strempty$-labelled path or there is a path:
    \[ \pi'' = s' \to^\strempty s'', \]
    In any case $\qcon(\pi'') = [ \open(q) ]$. We construct the combined path:
    \[ \pi' = \pi\pi''. \]
    and complete the proof for the case.
  \end{enumerate}
\item \label{enu:qgraph-build:gadget:pc:kpp1}
  When $k'' = 1$: The proof of this case is similar to
  Case~\ref{enu:qgraph-build:gadget:pc:kpp3:kp2} above.
  \qedhere
\end{enumerate}
\end{proof}

%% file: app/alg-qgraph-build-final.tex
\subsubsection{The Final Query Graph Construction}
\label{app:alg:qgraph-build:final}

We present the full proofs of Lemmas~\refshort{lem:alg:qgraph-build:qg-pc1} and~%
\refshort{lem:alg:qgraph-build:qg-pc2}.


\repthm{lem:alg:qgraph-build:qg-pc1}
\begin{proof}
By induction on $\pi$.
The base case follows from the presence of the path $(s_0, 1, 1) \to (s_0, 2, 1) \to (s_0, 3, 1)$ in
the first copy of the gadget in $G_M^w$.
In the inductive case, we are given a tentatively feasible path:
\[ \pi' = \underbrace{s_0 \to^* s}_\pi \to^a s', \]
for $a \in \Sigma \union \{ \strempty \}$. Because the prefix $\pi$ is also tentatively feasible, by
application of the induction hypothesis, we conclude the presence of a tentatively feasible path:
\[ p = \qgstart \to^* (s, 3, i), \]
for $i = |\str(\pi)|$. We visualize the situation in Figure~\ref{fig:alg:qgraph-build:qg-pc1:ind},
and consider the various cases:
\begin{enumerate}
\item When $a \in \Sigma$: Consider the following extended path:
  \[ p' = \underbrace{\qgstart \to^* (s, 3, i)}_p \to (s', 1, i + 1)
                                                  \to (s', 2, i + 1)
                                                  \to (s', 3, i + 1). \]
  The first new edge is drawn from $E_c$, and corresponds to the $s \to^a s'$ character transition.
  The remaining edges are intra-gadget connections, and are drawn from $E_\strempty$. Recall from
  Assumption~\refshort{assm:alg:qgraph-build:snfa2} that $\lambda(s') = \blank$, so that
  $l(s', 1, i + 1) = l(s', 2, i + 1) = l(s', 3, i + 1) = \blank$. It follows that $p'$ is
  also a tentatively feasible path with the same query context as $\pi'$.
\end{enumerate}
For the remaining cases, let $p$ be of the form:
\[ p = \underbrace{\qgstart \to^* (s'', 2, i)}_{p''} \to (s, 3, i). \]
\begin{enumerate}[resume]
\item When $a = \strempty$ and $\lambda(s') = \blank$: It follows that $s''$, $s$, and $s'$ all have
  the same query context, and that either there is an $\strempty$-labelled feasible path from $s'''
  \to^* s$ with $s'' \to^\strempty s'''$ or $s'' = s$, so we can construct the alternative
  tentatively feasible path with the same query context as the extended SNFA path $\pi'$:
  \[ p' = \underbrace{\qgstart \to^* (s'', 2, i)}_{p''} \to (s', 3, i). \]
\item When $a = \strempty$ and $\lambda(s') = \open(q)$: Observe that the gadget must contain an
  edge $(s'', 2, i) \to (s', 2, i)$, so that we can redirect path $p$ as follows:
  \[ p' = \underbrace{\qgstart \to^* (s'', 2, i)}_{p''} \to (s', 2, i) \to (s', 3, i). \]
  It can be shown that this path is tentatively feasible and has the same query context as $\pi'$.
\item When $a = \strempty$ and $\lambda(s') = \close(q)$: Here, we perform a second case split on
  $\lambda(s'')$:
  \begin{enumerate}
  \item If $\lambda(s'') = \blank$, or $\lambda(s'') = \close(q)$: Recall that all edges between the
    vertices of Layer~2 point towards $\open$ nodes. The predecessor of $(s'', 2, i)$ must therefore
    necessarily have been drawn from the first layer of the gadget, as follows:
    \[ p = \underbrace{\qgstart \to^* (s'', 1, i)}_{p'''} \to (s'', 2, i) \to (s, 3, i). \]
    In any case, the query context of $s'$ must be one query shorter than the query context of
    $s''$, and there must be a Layer~1 transition from $(s'', 1, i) \to (s', 1, i)$. We obtain the
    target path $p'$ by redirecting $p$ as follows:
    \[ p' = \underbrace{\qgstart \to^* (s'', 1, i)}_{p'''} \to (s', 1, i)
                                                           \to (s', 2, i)
                                                           \to (s', 3, i). \]
  \item If $\lambda(s'') = \open(q)$: In this case, the predecessor of $(s'', 2, i)$ must be present
    either in Layer~2 or in Layer~1. Suppose the predecessor is $(s''', 2, i)$, for some state
    $s'''$. Then $p$ looks as follows:
    \[ p = \underbrace{\qgstart \to^* (s''', 2, i)}_{p'''} \to (s'', 2, i) \to (s, 3, i). \]
    The query context of $s'$ must be the same as $s'''$ and there must be a Layer~2 to Layer~3
    transition $(s''', 2, i) \to (s', 3, i)$ . And we construct the target path $p'$ as follows:
    \[ p' = \underbrace{\qgstart \to^* (s''', 2, i)}_{p'''} \to (s', 3, i). \]
    On the other hand, assume that the predecessor of $(s'', 2, i)$ is some vertex $(s''', 1, i)$.
    Because the Layer~1 to Layer~2 edges only involve direct connections, $s''' = s''$. But
    Assumption~\refshort{assm:alg:qgraph-build:snfa2} allows us to argue that $(s'', 1, i)$ cannot
    be reachable from $\qgstart$, leading to a contradiction. \qedhere
  \end{enumerate}
\end{enumerate}
\end{proof}


\begin{figure}
\centering
\scalebox{\scalefactor}{
\begin{tikzpicture}
  \node [state, initial] (s0) {$s_0$};
  \node (s0Nw) at ($(s0.north) + (-.25, .3)$) {};

  \node [state, right=4 of s0] (s) {$s$};

  \node [state, right=2.5 of s] (sp) {$s'$};
  \node (spNe) at ($(sp.north) + (0.25, .3)$) {};

  \draw [decorate, decoration={coil, aspect=0}, ->] (s0) -- (s) node [midway, above] {$\pi$};

  \path [->] (s) edge [above] node {$a \in \Sigma \union \{ \strempty \}$} (sp);
  \draw [decorate, decoration={brace}] (s0Nw) -- (spNe) node [midway, above] {$\pi'$};

  \node [qgstart, below=of s0, label=below:{$(s_0, 1, 0)$}] (qgstart) {$\qgstart$};
  \node [below=of s] (qgs) {$(s, 3, i)$};
  \draw [decorate, decoration={coil, aspect=0}, ->] (qgstart) -- (qgs) node [midway, below] {$p$};

  \path [dashed] (s0) edge (qgstart);
  \path [dashed] (s) edge (qgs);
\end{tikzpicture}
}

\Description{Setting for the inductive step in the proof of
  Lemma~\refshort{lem:alg:qgraph-build:qg-pc1}.}
\caption{Setting for the inductive step in the proof of
  Lemma~\refshort{lem:alg:qgraph-build:qg-pc1}. We are given a path $\pi$ through the SNFA for which
  the conclusion is known to hold, and need to establish the result for a one step extension:
  $\pi' = s_0 \to^* s \to^a s'$. Both $\pi$ and $\pi'$ are tentatively feasible. The induction
  hypothesis promises the presence of a tentatively feasible path $p = (s_0, 1, 0) \to^* (s, 3, i)$
  through $G_M^w$, for $i = |\str(\pi)|$. We need to establish the existence of an extended path
  corresponding to $\pi'$.}
\label{fig:alg:qgraph-build:qg-pc1:ind}
\end{figure}


\mclearpage
\repthm{lem:alg:qgraph-build:qg-pc2}
\begin{proof}
By induction on $i$. The base case, when $i = 0$ follows immediately from Lemma~%
\ref{lem:qgraph-build:gadget:pc}. We now establish the inductive case.
In this case, we have to prove the result for a combined path:
\[ p' = \underbrace{\qgstart \to^* (s, 3, i)}_p \to
        \underbrace{(s', 1, i + 1) \to^* (s'', 3, i + 1)}_{p''}, \]
assuming that the result holds for the prefix $p$.

From the induction hypothesis, it follows that there exists a tentatively well-parenthesized path
$\pi$ in the SNFA with the same query context which is equifeasible with $p$. From Lemma~%
\ref{lem:qgraph-build:gadget:pc} we declare the existence of a tentatively feasible path $\pi''$
through the SNFA that has the same query context as $p''$. Also since there exists an inter-gadget
edge between $v$ and $v'$, there is a character transition $s \to^a s'$ in the SNFA. We construct the
combined path:
\[ \pi' = \pi \to^a \pi''. \]
The two paths, $p'$ and $\pi'$ are equifeasible. This completes the proof.
\end{proof}

%% file: app/alg-analysis.tex
\subsection{Performance Analysis of Our Algorithm}
\label{app:alg:analysis}

\repthm{thm:alg:analysis:main}
\begin{proof}
First, the SNFA $M_r$ can be constructed in time $O(|r|)$, and recall from Lemma~%
\refapp{lem:alg:qgraph-build:eps} that $\strempty$-feasibility can be determined in $O(|r|^2)$ time.
Thus the full query graph can be constructed in $O(|r|^2 |w|)$ time.

We also observe that $G_{M_r}^w$ has at most $O(|r||w|)$ vertices. From the construction of the
final query graph in Equation~\ref{eq:alg:graph-build:final} and the gadget in Equation~%
\ref{eq:alg:qgraph-build:gadget}, observe that the in-degree of any node is bounded by $O(|r|)$.
Also note that, because all vertices $v \in \backref(v)$ must be opening the same query $q_v$, it
must be the case that $|\backref(v)| \leq |w|$. For similar reasons, $|\aq(v)| \leq |w|$ and
$|\loq(v)| \leq |w|$. When $r$ is non-nested, $\loq(v) = \emptyset$. Finally, when $r$ is
unambiguous, $|\backref(v)|, |\aq(v)|, |\loq(v)| \leq 1$. 

Now, we catalogue the time needed to apply each rule to each vertex of the query graph in Table~%
\ref{tab:alg:analysis}. The promised time complexity may be obtained by multiplying the per-node
running times with the number of vertices $O(|r| |w|)$ in $G_{M_r}^w$.

Oracle calls are only made while computing $\aq(v)$. Because $|\aq(v)| \leq |w|$, at most $|w|$
oracle calls are made while computing $\aq(v)$. Overall, then, the algorithm makes no more than
$O(|r| |w|^2)$ oracle calls in the general and non-nested cases. In the unambiguous setting, because
$\backref(v')$ contains at most one element, it follows that at most one oracle query may be
discharged at each vertex query graph, which consequently limits the total number of discharged
queries to $O(|r| |w|)$.
\end{proof}

\begin{table}
\caption{Time needed to apply each rule from Figure~\refshort{fig:alg:qgraph-eval} to each vertex of
  the query graph, in the general case and for the case of unambiguous SemREs. Because the query
  graph contains $O(|r| |w|)$ vertices, the total running time is bounded by $O(|r|^2 |w|^2 +
  |r| |w|^3)$ in the general case, and by $O(|r|^2 |w|^2)$ and $O(|r|^2 |w|)$ in the non-nested and
  unambiguous cases respectively.}
\label{tab:alg:analysis}
\centering
\begin{tabular}{llll}
\toprule
  Rule         & Time (General)         & Time (Non-nested) & Time (Unambiguous) \\
\midrule
  \textsc{As}  & $O(1)$                 & $O(1)$            & $O(1)$ \\
  \textsc{Ab}  & $O(|r|)$               & $O(|r|)$          & $O(|r|)$ \\
  \textsc{Ao}  & $O(|r|)$               & $O(|r|)$          & $O(|r|)$ \\
  \textsc{Ac}  & $O(1)$                 & $O(1)$            & $O(1)$ \\
\midrule
  \textsc{M}   & $O(|r| |w|)$           & $O(|r| |w|)$      & $O(|r|)$ \\
\midrule
  \textsc{Bb}  & $O(|r| |w|)$           & $O(|r| |w|)$      & $O(|r|)$ \\
  \textsc{Bo}  & $O(1)$                 & $O(1)$            & $O(1)$ \\
  \textsc{Bc}  & $O(|w|^2)$             & $O(|w|)$          & $O(1)$ \\
\midrule
  $\loq(v)$    & $O(|r| |w|)$           & $O(|r|)$          & $O(|r|)$ \\
\midrule
  \textbf{Total} & $O(|r| |w| + |w|^2)$ & $O(|r| |w|)$      & $O(|r|)$ \\
\bottomrule
\end{tabular}
\end{table}

\begin{note}
\begin{enumerate}
\item Instead of the auxiliary quantity $\loq(v)$ and the Rule~\textsc{Bc} in
  Figure~\refshort{fig:alg:qgraph-eval}, an early version of our paper instead inlined the necessary
  reasoning into the conceptually simpler Rule~\textsc{Bcs} (the subscript ``\emph{s}'' indicating
  ``\emph{simple}'') shown in Figure~\ref{afig:alg:analysis:Bcs}. Despite its simplicity, the lack
  of caching provided by the explicit definition of $\loq(v)$ results in a poorer overall time
  complexity.

\item An early implementation of our algorithm explicitly constructed the query graph $G_M^w$.
  Although this implementation satisfied the asymptotic guarantees of
  Theorem~\refshort{thm:alg:analysis:main}, the repeated construction and discarding of query graphs
  for each line of text in the dataset turned out to greatly slow down the algorithm in practice.
  Our current implementation instead only constructs functions for the adjacency relation that allow
  us to traverse the vertices of $G_M^w$ as needed. This massively reduces memory pressure and
  speeds up the implementation.

\item Another optimization in our implementation is to compute the intermediate quantities,
  $\alive(v)$, $\aq(v)$, and $\backref(v)$ lazily in an on-demand manner. This
  significantly reduced the number of oracle calls made for strings which match.
\end{enumerate}
\end{note}

\begin{figure}
\centering
\scalebox{0.75}{$
    \inferrule*[Left=Bcs]
               {l(v) = \close(q) \\
                \alive(v) \\
                v' \in \aq(v) \\
                v'' \to v}
               {\backref(v'') \subseteq \backref(v)}
$}
\caption{Alternative query graph evaluation Rule~\textsc{Bcs}. This formulation inlines the
  definition of $\loq(v)$ and is therefore conceptually simpler. On the other hand, we require
  $O(|r| |w|^2)$ time to apply this rule to each vertex, which increases the overall time complexity
  to $O(|r|^2 |w|^3)$.}
\label{afig:alg:analysis:Bcs}
\end{figure}

%% file: app/complexity.tex
\section{The Complexity of Matching Semantic Regular Expressions}
\label{app:complexity}

\repthm{thm:complexity:quad}
\begin{proof}
When $Q$ is unbounded, construct the family of SemREs $r_n = \Sigma^* (\query{q_1} + \query{q_2} +
\dots + \query{q_n}) \Sigma^*$ and the family of strings $w_m = 0^m 1^m$, indexed by $n, m \in
\mathbb{N}$ respectively. We claim that any algorithm which determines whether $w_m \in
\interp{r_n}$ must evaluate $\oracle(q_i, 0^j 1^k)$ for each $i \in \{ 1, 2, \dots, n \}$ and for
each $0 \leq j, k \leq m$.

Assume, for the sake of contradiction, that there is an algorithm $A$ which determines the result
without evaluating $\oracle(q_i, 0^j 1^k)$. Consider the pair of oracles $\oracle_f$ and $\oracle_t$
defined as:
\begin{alignat*}{1}
  \oracle_f(q, w) & = \false, \text{ and} \\
  \oracle_t(q, w) & = \begin{cases}
                      \true  & \text{if } q = q_i \text{ and } w = 0^j 1^k, \text{ and} \\
                      \false & \text{otherwise},
                      \end{cases}
\end{alignat*}
respectively. Notably, the two oracles only produce different outputs when queried on the input
$(q_i, 0^j 1^k)$. Furthermore, observe that $w_m \notin \interp{r_n}$ according to $\oracle_f$ and
that $w_m \in \interp{r_n}$ according to $\oracle_t$, contradicting the assumption that $A$ can
correctly solve the matching problem.

The proof of the second part is similar to the first, except using the fixed SemRE $r = \Sigma^*
\query{q} \Sigma^*$ for some query $q \in Q$ instead of the family of SemREs, $r_n$.
\end{proof}

%% file: app/eval.tex
\section{Experimental Evaluation}
\label{app:eval}

\paragraph{Experimental Setup.}

We ran our experiments on a workstation machine with an AMD Ryzen 9 5950X CPU and 128~GB of memory
running Ubuntu~22.04. The workstation also has a Nvidia RTX 4080 GPU with 16~GB of GRAM for locally
running the LLM.